\documentclass[conference]{IEEEtran}
\IEEEoverridecommandlockouts

\usepackage{color}

\usepackage{tabularx}
\usepackage{tikz}
\usepackage{soul}
\usepackage{multirow}
\usepackage{cite}
 
\usepackage{graphicx}
\usepackage{float}

\usepackage{graphicx}
\usepackage[caption=false,font=footnotesize]{subfig} 

\usepackage{amsmath}

\usepackage{amsmath,amsfonts}
\usepackage{algorithm}
\usepackage{array}
\usepackage{algpseudocode}

\usepackage{textcomp}
\usepackage{stfloats}
\usepackage{url}
\usepackage{verbatim}
\usepackage{graphicx}
\usepackage{pifont}
\usepackage{makecell}
\usepackage{bm}
\usepackage{bbm}
\usepackage[english]{babel}
\usepackage{amsthm}
\usepackage{color}

\theoremstyle{definition}
\newtheorem{theorem}{Theorem}

\newtheorem{proposition}{Proposition}
\newtheorem{assumption}{Assumption}
\newtheorem{lemma}{Lemma}

\IEEEoverridecommandlockouts

\usepackage{tabularx}
\usepackage{tikz}
\usepackage{soul}
\usepackage{multirow}
 
\usepackage{graphicx}
\usepackage{float}

\begin{document}

\title{\vspace{-0.68\baselineskip}Energy Efficient Federated Learning
with Hyperdimensional Computing
over Wireless Communication
Networks\vspace{-0.28\baselineskip}
}

\author{Yahao Ding, Yinchao Yang, Jiaxiang Wang,
Zhaohui Yang,\\ Zhu Han, \IEEEmembership{Fellow, IEEE},
and Mohammad Shikh-Bahaei,
\IEEEmembership{Senior Member, IEEE}
\thanks{Yahao Ding, Yinchao Yang, Jiaxiang Wang, and Mohammad Shikh-Bahaei are with the Department of Engineering, King's College London, London, UK (emails: yahao.ding@kcl.ac.uk; yinchao.yang@kcl.ac.uk; jiaxiang.wang@kcl.ac.uk; m.sbahaei@kcl.ac.uk).}
\thanks{Zhaohui Yang is with the College of Information Science and Electronic Engineering, Zhejiang University, Hangzhou, Zhejiang 310027, China, and Zhejiang Provincial Key Lab of Information Processing, Communication and Networking (IPCAN), Hangzhou, Zhejiang, 310007, China (email: yang\_zhaohui@zju.edu.cn).}
\thanks{Zhu Han is with the Department of Electrical and Computer Engineering at the University of Houston, Houston TX 77004, USA, and also with the Department of Computer Science and Engineering, Kyung Hee
University, Seoul 446701, South Korea (e-mail: hanzhu22@gmail.com).}
}

\markboth{Journal of \LaTeX\ Class Files,~Vol.~14, No.~8, August~2015}%
{Shell \MakeLowercase{\textit{et al.}}: Bare Demo of IEEEtran.cls for IEEE Journals}

\maketitle
\vspace{-1\baselineskip}
\begin{abstract} 
In this paper, we investigate a problem of minimizing total energy consumption for secure federated learning (FL) over wireless edge networks. To address the high computational cost and privacy challenges in conventional FL with neural networks (NN) for resource-constrained users, we propose a novel FL with hyperdimensional computing and differential privacy (FL-HDC-DP) framework. In the considered model, each edge user employs hyperdimensional computing (HDC) for local training, which replaces complex neural updates with simple hypervector operations, and applies differential privacy (DP) noise to protect transmitted model information. We optimize the total energy of computation and communication under both latency and privacy constraints. We formulate the problem as an optimization that minimizes the total energy of all users by jointly allocating HDC dimension, transmission time, system bandwidth, transmit power, and CPU frequency. To solve this problem, a sigmoid-variant function is proposed to characterize the relationship between the HDC dimension and the convergence rounds required to reach a target accuracy. Based on this model, we develop two alternating optimization algorithms, where closed-form expressions for time, frequency, bandwidth, and power allocations are derived at each iteration. Since the iterative algorithm requires a feasible initialization, we construct a feasibility problem and obtain
feasible initial resource parameters by solving a per round transmission time minimization problem.
Simulation results demonstrate that the proposed FL-HDC-DP framework achieves up to 83.3\% total energy reduction compared with the baseline, while attaining about 90\% accuracy in approximately 3.5$\times$ fewer communication rounds than the NN baseline.

\end{abstract}


\IEEEpeerreviewmaketitle

\section{Introduction}



In the contemporary digital era, the proliferation of billions of intelligent edge devices, such as smartphones, Internet of things (IoT) devices, and wearable technology, is generating an unprecedented volume of data at the network edge. This trend has catalyzed the rise of edge artificial intelligence (AI), a paradigm focused on processing data and delivering intelligent services locally on these devices. To leverage collective intelligence of this decentralized data while preserving user privacy, federated learning (FL) has emerged as a key enabling technology \cite{10637245,9705093}. By training models locally on each device and sharing only anonymized model updates, FL obviates the need to transfer sensitive raw data to a central server, representing a transformative shift in privacy-preserving collaborative machine learning (ML) \cite{9738815,10353003,10660465}.

However, despite FL great promise, two core challenges constrain its large-scale, real-world deployment. The first is the fundamental bottleneck of energy consumption. Edge devices are typically resource-constrained and battery-powered, yet FL process is inherently energy-intensive due to both local computation, especially when training complex neural network (NN) models, and wireless communication. The second is the core concern of privacy leakage \cite{khaleghi2020prive}. Standard FL is not completely secure, as model updates are still vulnerable to inference attacks that could reveal sensitive user information, necessitating robust, quantifiable privacy guarantees.

To address the high computational energy challenge posed by complex models like NN, hyperdimensional computing (HDC) offers a compelling solution as an emerging, brain-inspired computing paradigm. HDC operates by representing and manipulating information in a vector space with thousands of dimensions. Unlike NN model, which is predicated on iterative backpropagation and floating-point arithmetic, HDC relies on a set of simple and computationally lightweight vector operations. These core operations include bundling, binding, and permutation, which are highly parallelizable and hardware-friendly. Concurrently, to address the privacy leakage risks in FL, differential privacy (DP) has been established as the canonical approach for providing rigorous, mathematically provable privacy guarantees by adding controlled noise.

Existing research has explored these challenges from various angles, yet a unified framework that co-optimizes all key metrics remains elusive. In the domain of resource efficiency optimization, the research has bifurcated into optimizing traditional FL and introducing novel computational paradigms. For traditional frameworks, the efforts have focused on minimizing energy consumption. For instance, the work in \cite{9264742} proposed a comprehensive joint optimization problem, collaboratively adjusting transmission time, power, bandwidth, and computation frequency. The authors of \cite{9916128} similarly focused on minimizing total energy by jointly optimizing weight quantization and wireless resource allocation. Other works have targeted reducing communication overhead and latency. For example, the authors of \cite{hsieh2021fl} and \cite{10473907} demonstrated improved communication efficiency through bipolarized hypervectors and the novel HyperFeel framework, respectively. Furthermore, the study in \cite{10925099} further combined federated split learning with HDC, designing an algorithm to minimize the maximum user transmission time.

The research in privacy and security is equally critical. To defend against threats like gradient leakage attacks, the work in \cite{10660465} proposed a scheme to minimize total privacy leakage by jointly optimizing resource allocation, the selection of layers for uploading, and DP noise. HDC introduces its own unique security considerations. The work in \cite{khaleghi2024private} thoroughly investigated security risks caused by reversible encoding of HDC and designed a framework with sparse encoding and a DP-enhanced model to achieve efficient and secure private learning. Furthermore, balancing privacy and accuracy in continuous learning is a key challenge. The FedHDPrivacy framework \cite{piran2025privacy} was designed specifically for this, actively monitoring and adding only the necessary noise to maintain model accuracy over multiple rounds.

More advanced and systemic designs have also emerged. The HyDREA system \cite{morris2022} focused on noisy FL scenarios, achieving robustness and high efficiency by adaptively adjusting model bitwidth based on the signal-to-noise ratio. From a hardware perspective,  the authors of \cite{9813404} addressed the energy efficiency bottleneck in HDC processors by generating hypervectors on-the-fly instead of relying on memory storage. More broadly, systemic studies and future-looking frameworks are being developed. For example, the authors of \cite{zhang2023hyperdimensional} conducted a systematic study on aggregation strategies and parameter impacts in HDC-FL. The study in \cite{xu2025new} explored the fusion of large AI models and HDC within the context of integrated learning and communication (ILAC) for future sixth generation (6G) networks, demonstrating how to balance learning performance with communication efficiency.

Therefore, a clear research gap exists: there is a lack of a unified framework
that jointly models and co-optimizes the computational efficiency of HDC, the
privacy guarantees of DP, as well as system-level computation and communication resources
to fundamentally solve the energy consumption problem for FL in wireless edge
environments. To fill this critical gap, this paper proposes a novel framework, FL with HDC and DP (FL-HDC-DP).  We specifically design this framework for energy-constrained wireless edge networks, enabling the users to leverage computational efficiency of HDC for local model training while obtaining robust privacy protection through mathematical rigor of DP. The central objective of our work is to minimize the total energy consumption of all participating users via a joint optimization strategy. The contributions of the paper can be summarized as follows:
\begin{itemize}
\item We propose an FL-HDC-DP framework designed for edge devices in wireless networks, where the users employ HDC instead of traditional NN for local
model training. We analyze effects of HDC dimension, DP noise, and data distribution (independent and identically distributed (IID) vs. non-IID) on model
performance.



\item We formulate a joint optimization problem aimed at minimizing total energy consumption, encompassing both local computation and wireless communication. To the best of our knowledge, this is the first work to jointly co-optimize the HDC model dimension together with system resources, including transmission time, transmit power, bandwidth, and CPU frequency. To solve this problem, we first develop a sigmoid-variant model to capture the relationship between HDC dimension and the required convergence rounds. Based on the formulated model, we solve the problem via an alternating algorithm, where the closed-form solution is obtained at each iteration.




\item To obtain a feasible solution for the total energy minimization problem, we construct a per round transmission time minimization problem. Starting from a given initial HDC dimension, a closed-form
feasibility criterion is proposed using the maximum admissible per round communication time and the minimum required bandwidth. If feasible, the initial bandwidth allocation is obtained by solving a convex transmission time minimization via a bisection search over the associated Lagrange multiplier, followed
by closed-form computation of the corresponding transmission times and CPU frequencies.

\item Simulation results indicate that the proposed FL-HDC-DP achieves approximately 90\% accuracy while converging about five times faster than the NN baseline (FL-NN-DP); equivalently, to reach the same accuracy, FL-NN-DP requires about $3.5\times$ as many epochs as FL-HDC-DP. Moreover, the proposed optimization method, which jointly optimizes model dimension, transmission time, power, and CPU frequency, reduces total energy consumption by up to 83.3\% compared to the baseline.
\end{itemize}

The rest of the paper is organized as follows. The preliminary concepts of the HDC model and DP are illustrated in Section \ref{HDC}. The system model is described in Section \ref{model}. The optimization problem is formulated in Section \ref{optimization}, and solution processes are given in Section \ref{AD}. Section \ref{fea} provides a method to find a feasible original solution of the proposed optimization problem. Section \ref{simulation} describes the simulation results and analysis of the proposed approach. Finally, the conclusions are drawn in Section \ref{con}.






\vspace{-0.5em}
\section{Preliminaries}\label{HDC}

\subsection{Notions of HDC}
HDC is a brain-inspired computing paradigm that 
represents information using a high-dimensional hypervector (HV)~\cite{kanerva2009hyperdimensional}. Its learning and inference mainly rely on simple, dimension-preserving operations together with similarity-based retrieval, which yields fast training, hardware-friendly computation, and robustness to noise. An HV is a $d$-dimensional vector with $d$ typically on the order of $10^3$--$10^4$ \cite{ge2020classification}, and we use bipolar HVs, i.e., $\boldsymbol{H}\in\{+1,-1\}^{d}$. Due to the high dimensionality, randomly generated HVs are quasi-orthogonal with high probability \cite{ma2024hyperdimensional}, providing a stable basis for encoding and comparison. HDC manipulates HVs via three key algebraic operations: bundling (addition) aggregates multiple HVs into a single HV to represent collective information, binding (multiplication) associates two HVs to encode their relationship into a new HV, and permutation applies a fixed reordering to inject positional or sequential structure while maintaining similarity-related properties. The framework utilizes two memory components: an item memory (IM) that stores fixed, orthogonal HVs for basic symbols, and an associative memory (AM) that stores class-specific prototypes. During inference, the query HV is compared against these prototypes to determine the nearest match for classification.

\subsection{HDC Model Development}
\begin{figure*}[t]
\vspace{-2.08em}
\centering\includegraphics[width=0.95\textwidth]{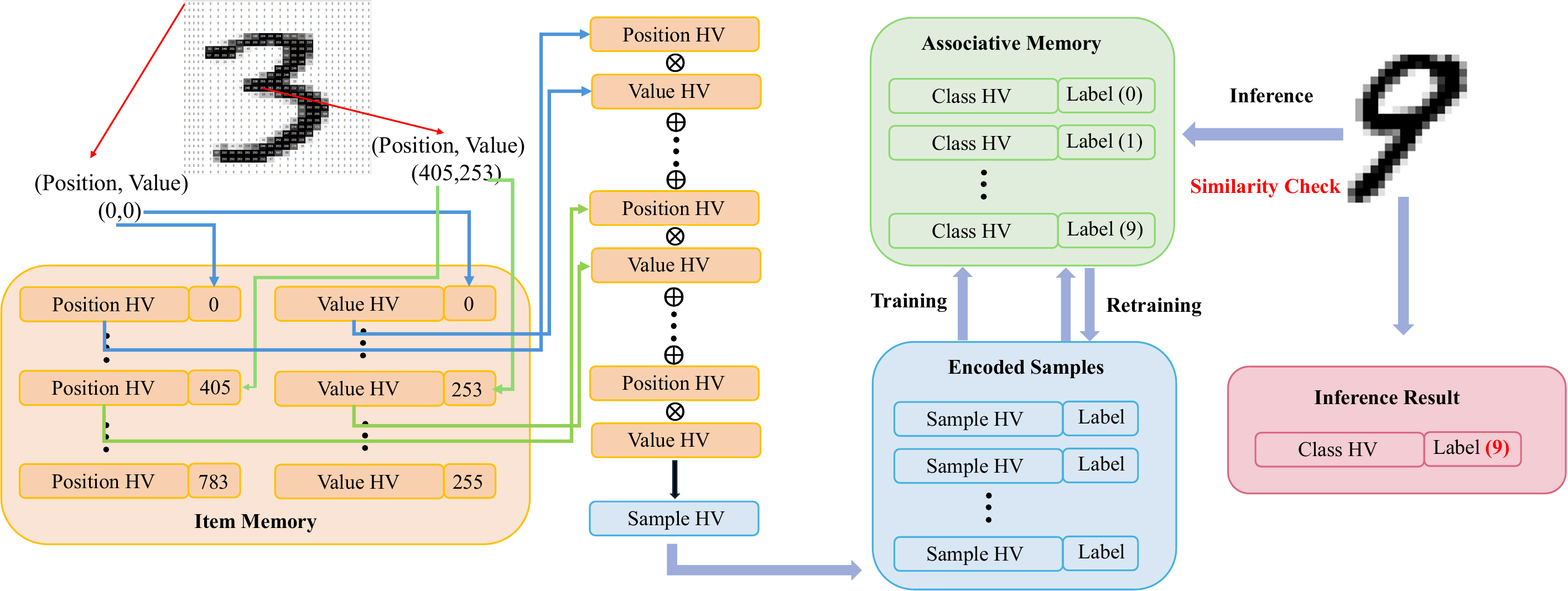}
\vspace{-0.68em}
\caption{Overview of HDC model encoding, training, inference, and retraining.}
\label{fig:H1}
\vspace{-1.28em}
\end{figure*}

HDC model development consists of encoding, training, inference, and optional retraining. The detailed processes are illustrated in Fig. \ref{fig:H1}. We adopt record-based encoding \cite{rahimi2016hyperdimensional} to map an input sample with $S$ pixels $\{(s,a_s)\}_{s=1}^{S}$ into a $d$-dimensional bipolar HV. Specifically, a position item memory assigns each pixel position $s$ a random HV, denoted by position HV $\boldsymbol{P}_s\in\{+1,-1\}^{d}$, while a value item memory maps the pixel value $a_s$ to a value HV $\boldsymbol{V}(a_s)\in\{+1,-1\}^{d}$. The sample HV is then formed by binding position and value HVs at each position and bundling across all positions:
\vspace{-0.48em}
\begin{equation}
\boldsymbol{H}=\sum_{s=1}^{S}\Big(\boldsymbol{P}_s \odot \boldsymbol{V}(a_s)\Big), \vspace{-0.48em}
\label{encoding}
\end{equation}
where $\odot$ denotes the Hadamard product. During training, class HVs $\boldsymbol A_n$ are obtained by adding all encoded HVs of the same class $n$, as given by
\vspace{-0.68em}
\begin{equation}\label{eq:hdc_training}
\boldsymbol{A}_n=\sum_{\ell:\,y_{\ell}=n}\boldsymbol{H}_{\ell},\vspace{-0.48em}
\end{equation}
where $\boldsymbol{H}_{\ell}$ denotes the encoded HV of the $\ell$-th training sample and $y_{\ell}\in\{1,\dots,N\}$ is its ground-truth class label. All class HVs constitute the AM $\mathcal{A}=\{\boldsymbol{A}_1,\dots,\boldsymbol{A}_N\}.$

For inference, a test sample is encoded into a query HV $\boldsymbol{H}_q$, and the predicted label is determined by maximum similarity,
\vspace{-0.48em}
\begin{equation}\label{eq:hdc_inference}
\hat{y}=\arg\max_{n\in\{1,\dots,N\}}\ \mathrm{sim}\!\left(\boldsymbol{H}_q,\boldsymbol{A}_n\right),\vspace{-0.48em}
\end{equation}
where $\mathrm{sim}(\cdot,\cdot)$ can be cosine similarity or Hamming distance. To further improve the accuracy of the HDC model, retraining can be applied to the AM over several additional iterations on the training set. If the prediction is correct, no adjustment is necessary. If the prediction is incorrect, i.e., $\hat{y}\neq y^\star$ where $y^\star$ is the ground-truth class, the AM can be refined by weakening the predicted class and strengthening the true class:
\vspace{-0.48em}
\begin{equation}
\begin{array}{r}
\boldsymbol{A}_{\hat{y}}= \boldsymbol{A}_{\hat{y}}-\eta\,\boldsymbol{H}_q, \\
\boldsymbol{A}_{y^\star}=\boldsymbol{A}_{y^\star}+\eta\,\boldsymbol{H}_q,
\end{array}
\label{retraining}\vspace{-0.48em}
\end{equation}
where $\eta$ is the learning rate.



\subsection{Preliminary Concepts of DP and zCDP}
DP provides a rigorous framework for quantifying privacy loss \cite{dwork2006calibrating}. A randomized mechanism $\mathcal{M}$ satisfies $(\epsilon, \delta)$-DP if for adjacent datasets $D_1, D_2$:
\vspace{-0.48em}
\begin{equation}\label{eq:dp_def}
\Pr\big[\mathcal{M}(D_1)\in L\big]\le e^{\epsilon}\Pr\big[\mathcal{M}(D_2)\in L\big]+\delta,\vspace{-0.48em}
\end{equation}
where $\epsilon$ is the privacy budget and $\delta$ is the failure probability. The Gaussian mechanism achieves this by adding noise proportional to the $\ell_2$-sensitivity, $\Delta = \max_{D_1 \sim D_2} \|\mathcal{M}(D_1) - \mathcal{M}(D_2)\|_2$. However, for multiple communication rounds, standard $(\epsilon, \delta)$ composition is loose. We therefore adopt $\rho$-zCDP, which yields tighter bounds. A mechanism is $\rho$-zCDP if the Rényi divergence between outputs on adjacent datasets is bounded linearly by $\rho$ \cite{bun2016concentrated}. A key advantage of zCDP is its additive composition: for $J$ rounds, the total privacy cost is $\rho_{\text{tot}} = \sum_{j=1}^J \rho_j$.We calibrate the Gaussian noise variance for zCDP using $\sigma^2 = \Delta^2 / (2\rho)$. The final cumulative guarantee can be converted back to standard $(\epsilon, \delta)$-DP via:
\vspace{-0.48em}
\begin{equation}
\varepsilon(\delta)\;=\;\rho_{\text{tot}}+2\sqrt{\rho_{\text{tot}}\ln\!\big(1/\delta\big)}.
\end{equation}


\section{System Model}\label{model}


We consider a single-cell network with one base station (BS) serving $U$ users, as shown in Fig.~\ref{fig:HDC-Fl-DP1}. They cooperatively train an FL model over wireless networks for classification inference \cite{10925099}. 
To reduce computational burden on resource-constrained edge devices and to protect user privacy against an honest-but-curious server and external eavesdroppers, we adopt a secure FL framework that combines HDC with DP, named FL-HDC-DP. 
Let $\mathcal{U}=\{1,2,\ldots,U\}$ denote the user set. 
User $i\in \mathcal{U}$ holds a private dataset $\mathcal{D}_i=\{(\boldsymbol{x}_{i,\ell},y_{i,\ell})\}_{\ell=1}^{D_i}$ where $\boldsymbol{x}_{i,\ell}$ is the $\ell$-th input sample and $y_{i,\ell}\in\{1,\ldots,N\}$ is its class label. Each sample is encoded into a bipolar HV $\boldsymbol{H}_{i,\ell}\in\{+1,-1\}^{d}$ using the record-based encoder introduced in Section~II. Moreover, the main notations are listed in Table~\ref{tab:main_notations_ch6}.

\begin{figure}[!t]
\vspace{-1.58em}
\centering\includegraphics[width=0.45\textwidth]{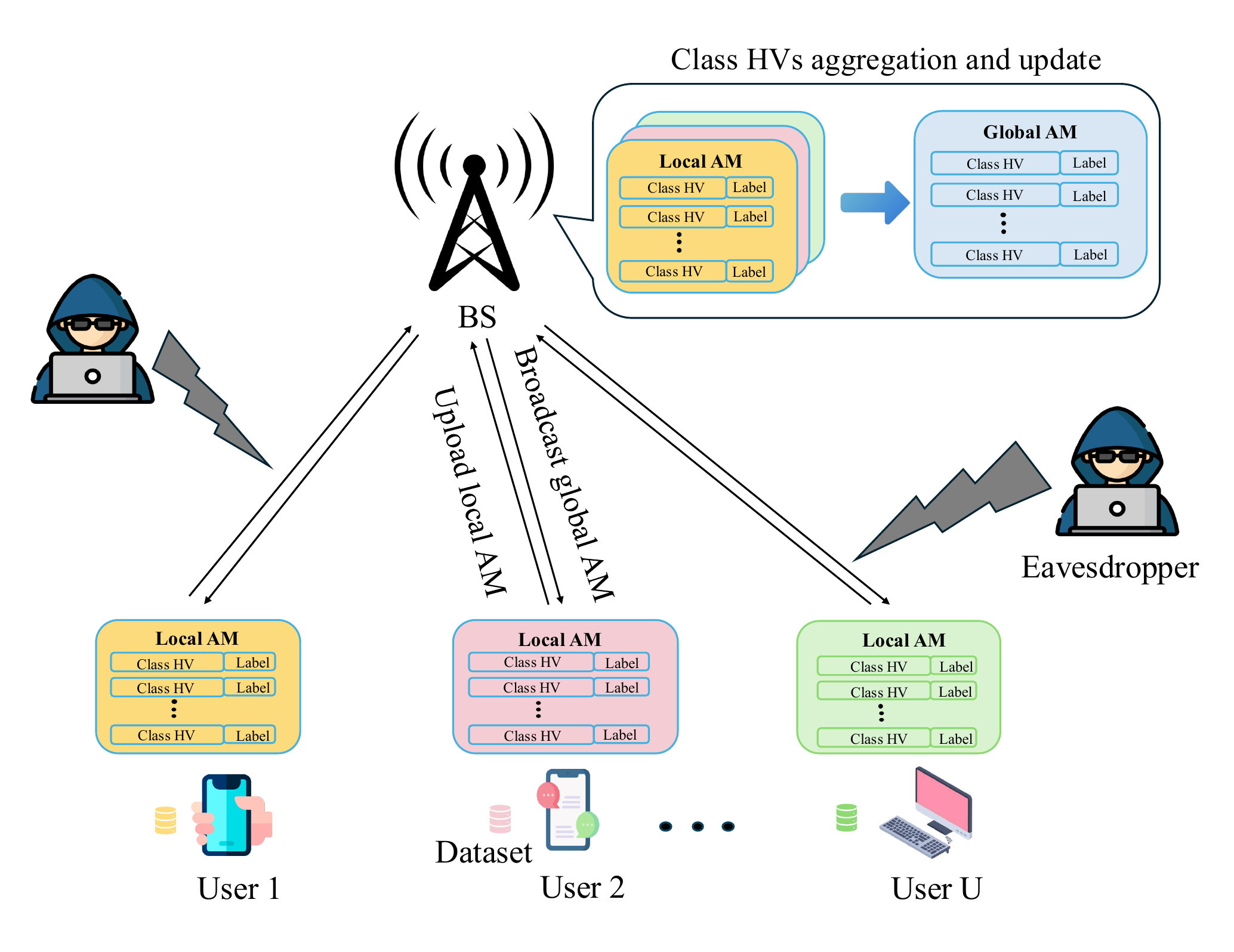} 
\vspace{-1.68em}
\caption{Illustration of the considered model for FL-HDC over wireless communication networks.}
\vspace{-1.28em}
\label{fig:HDC-Fl-DP1}
\end{figure}

\begin{table}[!t]
\centering
\caption{Main Notations}
\vspace{-0.68em}
\label{tab:main_notations_ch6}
\begin{tabular}{l|p{7cm}} 
\hline\hline
\textbf{Symbol} & \textbf{Description} \\
\hline
$\boldsymbol{A}_{i,n}$ & The $n$-th class HV held by user $i$ \\
$\mathcal{A}_{i}$ & The AM of user $i$ \\
$B$      & Total system bandwidth \\
$D_i$    & Number of data samples for user $i$ \\
$\mathcal D_{\text{HDC}}$    & Candidate set of HDC dimensions \\
$e_i$    & Inference error rate of user $i$ \\
$E$      & Total energy consumption of all users \\
$g_i$    & Channel power gain for user $i$ \\
$\boldsymbol{H}_{i,\ell}$ & $\ell$-th encoded HV of user $i$ \\
$J$      & Total number of communication rounds \\
$J_d$    & Rounds required for convergence with dimension $d$ \\
$\kappa$ & Clipping bound for HVs \\
$N_0$    & Noise power spectral density \\
$N_i$    & Number of classes in the AM of user $i$ \\
$P_{\max}$ & Maximum transmit power \\
$r_i$    & Uplink transmission rate of user $i$ \\
$\sigma_j$ & Standard deviation of DP noise in round $j$ \\
$T$      & Total completion time for FL training \\
$U$ & Total number of users \\
$\delta$   & Privacy loss threshold \\
$\epsilon$ & Privacy budget in DP \\
\hline
$d$      & HDC dimension \\
$b_i$    & Bandwidth allocation of user $i$ \\
$p_i$    & Transmit power of user $i$ \\
$f_i$    & CPU frequency of user $i$ \\
$t_i$    & Transmission time of user $i$ \\
\hline\hline
\end{tabular}
\end{table}

\subsection{DP Mechanism and Sensitivity in FL-HDC}



The FL-HDC-DP framework ensures privacy in FL by incorporating DP into the exchange of class HV between clients and the server. Without a privacy-preserving mechanism, this shared information is vulnerable to model inversion and membership inference attacks, which could reveal sensitive training data. To mitigate these risks, Gaussian noise is added to the HDC models, with the variance of the noise calculated based on key privacy parameters: the privacy budget $\epsilon$, the privacy loss threshold $\delta$, and the sensitivity of the function $\Delta$.

In our work, we apply zCDP \cite{bun2016concentrated} to FL-HDC. The $n$-th class HV held by user $i$ at round $j$ is denoted by $\boldsymbol{A}^{j}_{i,n}\in\mathbb{R}^{d}$, and AM of user $i$ is denoted as $\mathcal A^{j}_i=\{\boldsymbol{A}^{j}_{i,1},\ldots,\boldsymbol{A}^{j}_{i,N}\}$.
Our DP guarantee is defined with record-level add or remove adjacency, i.e., one sample may be added or deleted between adjacent datasets.

\begin{assumption}[HDC clipping]
For any user $i$ and sample $\ell$, the encoded bipolar HV is $\ell_2$-clipped before any local aggregation or communication \cite{10720539}:
\vspace{-0.48em}
\begin{equation}
\boldsymbol{H}^{\mathrm{clip}}_{i,\ell}
=\boldsymbol{H}_{i,\ell}\cdot
\min\left\{1,\ \frac{\kappa}{\lVert \boldsymbol{H}_{i,\ell}\rVert_2}\right\},
\big\lVert \boldsymbol{H}^{\mathrm{clip}}_{i,\ell}\big\rVert_2\le \kappa.
\label{clipping}\vspace{-0.48em}
\end{equation}
Before clipping, $\big\lVert\boldsymbol{H}_{i,\ell}\big\rVert_2=\sqrt{d}$ and after clipping equals $\kappa$.
\end{assumption}

Two conditions should be analyzed separately:
\begin{enumerate}
    \item One-pass training ($j=1$):
Each user performs a single pass of encoding and aggregation, summing clipped sample hypervectors with the same label to form class HVs:
\vspace{-0.48em}
\begin{equation}
\boldsymbol{A}^{1}_{i,n}
=\sum_{\ell:\,y_{i,\ell}=n}\boldsymbol{H}^{\mathrm{clip}}_{i,\ell}.\vspace{-0.48em}
\end{equation}
Under sample-level add or remove adjacency, only one sample differs between neighboring datasets; hence, the $\ell_2$-sensitivity of the local AM sent in round~$1$ is
\vspace{-0.48em}
\begin{equation}
\Delta^{1}=\kappa.\vspace{-0.48em}
\end{equation}

\item Retraining rounds ($j\ge 2$):
In round $j\ge 2$, user $i$ compares each $\boldsymbol{H}^{\mathrm{clip}}_{i,\ell}$ with the global AM. If a sample is misclassified with predicted label $\hat y_{i,\ell}\neq y_{i,\ell}$, user $i$ updates two classes:
\vspace{-0.48em}
\begin{equation}
\begin{aligned}
&\boldsymbol{A}^{j}_{i,y_{i,\ell}}= \boldsymbol{A}^{j}_{i,y_{i,\ell}} + \boldsymbol{H}^{\mathrm{clip}}_{i,\ell},\\
&\boldsymbol{A}^{j}_{i,\hat y_{i,\ell}} = \boldsymbol{A}^{j}_{i,\hat y_{i,\ell}} - \boldsymbol{H}^{\mathrm{clip}}_{i,\ell}.
\end{aligned}\vspace{-0.48em}
\end{equation}
A single sample therefore affects at most two class HVs whose norms are bounded by $\kappa$, which yields the per-round sensitivity
\vspace{-0.48em}
\begin{equation}
\Delta^{j}=\sqrt{2}\,\kappa,\qquad j\ge 2.\vspace{-0.48em}
\end{equation}
\end{enumerate}

Given the sensitivity, the per-round noise required to ensure that the overall $J$-round mechanism satisfies $(\varepsilon,\delta)$-DP is characterized by the following lemma.

\begin{lemma}[Per-round Gaussian noise to ensure $(\varepsilon,\delta)$-DP after $J$ rounds]
We calibrate noise using zCDP. Define
\vspace{-0.48em}
\begin{equation}
 \rho_{\max}=\Big(\sqrt{\varepsilon+\ln(1/\delta)}-\sqrt{\ln(1/\delta)}\Big)^{2}.   \vspace{-0.48em}
\end{equation}
Let the sample-level $\ell_{2}$-sensitivities be $\Delta^1=\kappa$ for the one-pass round and 
$\Delta^{j}=\sqrt{2}\,\kappa$ for retraining rounds $j=2,\ldots,J$. 
Then, to ensure that the overall $J$-round mechanism satisfies $(\varepsilon,\delta)$-DP, the required per-round \emph{noise standard deviation} added to each coordinate of every class HV $\boldsymbol{A}^{j}_{i,n}$ is
\vspace{-0.48em}
\begin{equation}
\sigma_j=
\begin{cases}
\displaystyle \kappa\,\sqrt{\dfrac{J}{2\,\rho_{\max}}}, & j=1,\\
\sqrt{2}\,\kappa\,\sqrt{\dfrac{J}{2\,\rho_{\max}}}, & j=2,\ldots,J.
\end{cases}
\label{DPnoise}\vspace{-0.48em}
\end{equation}
\end{lemma}

\noindent\emph{Proof.}
A Gaussian mechanism with $\ell_{2}$-sensitivity$\Delta$ and per-coordinate noise standard deviation $\sigma$ is $\rho=\Delta^{2}/(2\sigma^{2})$-zCDP. 
Across $J$ rounds, zCDP composes additively:
$\rho_{\text{tot}}=\sum_{j=1}^{J}{\Delta^{j}}^2/\big(2\sigma_j^{2}\big)$. 
With the choice $\sigma_j=\Delta^j\sqrt{J/(2\rho_{\max})}$, we obtain $\rho_{\text{tot}}=\rho_{\max}$. 
Finally, any $\rho$-zCDP mechanism implies $(\varepsilon,\delta)$-DP via $\varepsilon=\rho+2\sqrt{\rho\ln(1/\delta)}$. \hfill$\square$



\subsection{The Framework of FL-HDC-DP}
The details of the FL-HDC-DP training process include three steps and the overall procedure is summarized in Algorithm~\ref {alg:fl-hdc-dp}:

\textbf{Step 1: Server-side initialization.}
The BS first generates bipolar HV benchmarks to ensure consistency in the
position HVs and value HVs across all clients. It also initializes a global AM and broadcasts it to all users, together with the training hyperparameters such as learning rate, DP clipping bound, and the number of local passes per round.

\textbf{Step 2: Client-side local training with DP.}
In the first round, each user encodes all local raw data into sample HVs using Eq. \eqref{encoding}, and clips them according to Eq. \eqref{clipping}, and then aggregates them according to their class to obtain their local class HVs $\boldsymbol{A}_{i,n}$ and the local AM $\mathcal{A}_i$. In subsequent rounds, the users perform retraining guided by the received global AM $\mathcal{A}_g^j$ at the $j$th round from the server. Specifically, all local sample HVs are compared with each class HV in the global AM $\mathcal{A}_g^j$ using the cosine similarity measure. For incorrectly predicted cases, the class HVs are updated according to \eqref{retraining}.  
After encoding or retraining, the resulting $\mathcal{A}_{i}^{j+1}$ is added with DP noise as specified as follows:
\vspace{-0.48em}
\begin{equation}
\tilde{\mathcal{A}}_{i}^{j+1}= {\mathcal{A}}_{i}^{j+1} + \boldsymbol{n}_i^{j+1}. \vspace{-0.48em}
\end{equation}
The noise standard deviation of $\boldsymbol{n}_i^{j+1}$ is calculated according to Eq. \eqref{DPnoise}. Finally, local AMs are transmitted to the BS. 



\textbf{Step 3: Server aggregation and broadcast.}
Upon receiving the DP-protected local AMs from all users in $\mathcal{U}$, the BS aggregates them by calculating the element-wise average of the class HV to obtain the updated global AM, as given by
\vspace{-0.48em}
\begin{equation}
{\mathcal{A}}_{g}^{j+1}=\frac{\sum_{i=1}^U{\tilde{\mathcal{A}}}_{i}^{j+1}}{U}\vspace{-0.48em}
\end{equation}
The updated global AM $\mathcal{A}_{g}^{j+1}$ is then broadcast to all users.

Steps 2–3 repeat until the desired accuracy is achieved or the global loss converges.


\begin{algorithm}[t]
\caption{FL-HDC-DP Algorithm}
\label{alg:fl-hdc-dp}
\begin{algorithmic}[1]
\Statex \textbf{Input:}User set $\mathcal{U}$; class set $\mathcal{N}$; user dataset $\mathcal{D}_i$, round $J$, learning rate $\eta$; clip bound $\kappa$; noise standard deviation\ $\{\sigma_j\}_{j=1}^{J}$.
\Statex \textbf{Output:} Final global AM $\mathcal{A}_g^J$. 
\For{$j = 0$ to $J-1$}
\State BS broadcasts global AM $\mathcal{A}_g^j$ and bipolar HV benchmarks to all users.
\For{user $i \in \mathcal{U}$ in parallel}
        \If{$j=0$}
        \State Encode local data to obtain $\boldsymbol{H}_{i,\ell}$ and clip them according to Eq.\eqref{clipping}.
        \State Aggregate all HVs to obtain local AM $\mathcal{A}_{i}$.

        \Else
            \State Perform retraining guided by $\mathcal{A}_g^j$ to update the local AM according to Eq.\eqref{retraining}, yielding $\mathcal{A}_i^{j+1}$.
        \EndIf 
        \State Add DP noise to $\mathcal{A}_i$: $\tilde{\mathcal{A}}_i^{j+1} =\mathcal{A}_i^{j+1} + \boldsymbol{n}_i^{j+1}$.
        \State Transmit $\tilde{\mathcal{A}}_i^{j+1}$ to the BS.
    \EndFor
    
    \Statex \textbf{Server-side Aggregation:}
    \State Receive $\{\tilde{\mathcal{A}}_i^{j+1}\}_{i=1}^U$ from all users.
    \State Conduct element-wise
average to update the global AM: $\mathcal{A}_g^{j+1} = \frac{1}{U} \sum_{i=1}^U \tilde{\mathcal{A}}_i^{j+1}$.
    \State Broadcast $\mathcal{A}_g^{j+1}$ to all users for the next round.
\EndFor
\end{algorithmic}
\end{algorithm}

\subsection{Transmission Model}
After local computing, the users upload their DP-added AMs to the BS for aggregation. We adopt frequency division multiple access (FDMA). The achievable uplink transmission rate between user \(i\) and the server with the allocated bandwidth is given by \cite{7762913}
\vspace{-0.48em}
\begin{equation}\label{eq:fdma_rate}
r_i \;=\; b_i \log_2\!\left(1+\frac{p_i g_i}{N_0 b_i}\right),\vspace{-0.48em}
\end{equation}
where \(N_0\) denotes the noise power spectral density, \(b_i\) is the bandwidth allocated to user \(i\), \(p_i\) is the user's transmit power, and \(g_i\) is the end-to-end channel power gain. The per-round bandwidth allocations satisfy
$\sum_{i\in\mathcal{U}} b_i \le B$ and $b_i\ge0$, where $B$ is the total bandwidth.
The corresponding transmission time is given by:
\vspace{-0.48em}
\begin{equation}\label{eq:t}
t_i \;=\; \frac{N_id}{r_i}.\vspace{-0.48em}
\end{equation}
where $N_id$ be the payload size of the AMs update for user \(i\), $N_i$ is the number of classes in the AM of user $i$, and $d$ is the HDC dimension.

\subsection{Energy Consumption Model}
In our network, we focus on the energy consumption of each user during the FL training phase, which can be divided into two primary components: a) local computation, including encoding the dataset and HDC model retraining by using its local dataset and the received global AM from server, and b) wireless communication, which involves transmitting the updated local AM.

\subsubsection{Energy Consumption of Computation Model} During the FL training process, in the first round, each user encodes the entire local dataset into HVs and aggregates HVs of the same class to construct the local AM. In the subsequent rounds, the user retrains the HDC model using the encoded HVs together with the global AM received from the server. These procedures are implemented through simple operations such as bundling, binding, and permutation.

Let \(C_{\rm enc}\), \(C_{\rm agg}\), \(C_{\rm sim}\), and \(C_{\rm up}\) denote the CPU cycles per dimension required for encoding, first‑round aggregation, similarity comparison, and error‑driven update, respectively. Let $D_i$ denote the total number of data sample for user $i$, 
$d$ is HV dimension,
$f_i$ is CPU frequency,
$\gamma$ is the switched‑capacitance coefficient,
$e_i = \frac{M_i}{D_i}$ is inference error rate, and $M_i$ is the number of misclassifications per round.
\begin{itemize}
    \item Round 1 (Encoding \& Aggregation): Define the per‑dimension initialization CPU cycles:
$
  C_i^{(\text{init})}
  = C_{\text{enc}} + C_{\text{agg}}
$, so that the latency and energy are given by\cite{7572018}
\vspace{-0.48em}
\begin{equation}
    \tau_i^{1} = \frac{D_i\,d\,C_i^{(\text{init})}}{f_i}, \quad \forall i \in \mathcal{U},\vspace{-0.48em}
\label{t1}
\end{equation}
\begin{equation}
    E_i^{1} = \gamma\,D_i\,d\,C_i^{(\text{init})}\,f_i^2.
\label{e1}\vspace{-0.48em}
\end{equation}
where $D_i\,d\,C_i^{(\text{init})}$ is the total CPU cycles.
 \item Rounds $J\ge2$ (Retraining): Define the per‑dimension retraining CPU cycles:
$
  C_i^{(\text{ret})}
  = C_{\text{sim}}+ e_i\,C_{\text{up}}.
$
Then, the time and energy consumption at user $i$ for retraining the HDC model can be expressed as follows:
\vspace{-0.48em}
\begin{equation}
    \tau_i^{(\text{ret})} = \frac{D_i\,d\,C_i^{(\text{ret})}}{f_i},  \quad \forall i \in \mathcal{U},\vspace{-0.48em}
\label{tr}
\end{equation}
\begin{equation}
   E_i^{(\text{ret})} = \gamma\,D_i\,d\,C_i^{(\text{ret})}\,f_i^2.\vspace{-0.48em}
\label{er}
\end{equation}
Eq. \eqref{er} captures the cost of similarity checks plus inference error‑driven updates, with inference error-rate \(e_i\).
\end{itemize}




\subsubsection{Energy Consumption of Transmission Model} After local computation, users transmit their local AM to the server for aggregation via FDMA. Based on Eqs. \eqref{eq:fdma_rate} and \eqref{eq:t}, the energy consumption for data transmission can be expressed as 
\vspace{-0.48em}
\begin{equation}
    E_i^{(\text{trans})}=t_i\,p_i,\vspace{-0.48em}
\end{equation}
where $p_i$ is the transmission power.

\subsubsection{Total Energy Consumption}
The total energy consumption of all users participating in FL is given as
\vspace{-0.48em}
\begin{align} 
E=&\!\!\!\underbrace{\sum_{i=1}^{U} \left(E_i^{1}\!+\!E_i^{(\text{trans})}\right)}_{\text{First round energy consumption}}\!\!\!\!\!+ \underbrace{(J\!-\!1)\sum_{i=1}^{U} \left(E_i^{(\text{ret})}\!+\!E_i^{(\text{trans})}\right)}_{\text{Energy consumption of}\,{J-1}\,\text{rounds}},\nonumber\\
=&\sum_{i=1}^U \!\gamma D_i d C_i^{(\text{init})}\!f_i^2 \!+ \!(J\!-\!1)\gamma D_i d C_i^{(\text{ret})}\!f_i^2 \!+\! J t_i p_i,
\label{E11}\vspace{-0.48em}
\end{align}
where $J$ is the total number of communication rounds. Since the operations in the first round differ from those in the subsequent rounds, the energy consumption is calculated separately. The former of Eq. \eqref{E11} represents the energy consumption of all users in the first round, while the latter denotes the energy consumption of all users in the remaining $J-1$ rounds.

Hereinafter, the total time needed for completing the execution of the FL algorithm is called the completion time. The completion time of each user includes the local computation time and transmission time, based on Eqs. \eqref{t1} and  \eqref{tr}, the completion time of user $i$ will be 
\vspace{-0.48em}
\begin{align} 
T_i=&\underbrace{\tau_i^{1}+t_i}_{\text{First round time cost}}+ \underbrace{(J-1) \left(\tau_i^{(\text{ret})}+t_i\right)}_{J-1\, \text{rounds time cost}},\nonumber \\ 
=&\frac{D_i d \left(C_i^{(\text{init})}+(J-1) C_i^{(\text{ret})}\right)}{f_i} + J \frac{N_id}{r_i}.
\end{align}

\section{Problem Formulation}\label{optimization}
In this section, we formulate the total energy consumption minimization problem of all users for secure FL-HDC, which jointly considers HDC dimension, transmission time, bandwidth, transmit power, and computing frequency. The optimization problem is given by
\vspace{-0.5em}
\begin{align}
&\min_{_{d,\boldsymbol{t},\boldsymbol{b},\boldsymbol{p},\boldsymbol{f}}}  \quad E\label{A}\\
\text{s.t.}
&\frac{D_i d \left(\!C_i^{(\text{init})}\!+\!(J_d\!-\!1) C_i^{(\text{ret})}\!\right)}{f_i} \!+ \!J_d t_i\!\leq \!T,\!\quad\! \forall i\in \mathcal U \tag{\ref{A}{a}}, \label{Aa}\\
&t_ib_i\log_2\left(1+\frac{p_ig_i}{N_0b_i}\right)\geq N_id,\quad\forall i\in  \mathcal U \tag{\ref{A}{b}}, \label{Ab}\\
& \sum_{i=1}^{K} b_i\leq B \tag{\ref{A}{c}}, \label{Ac}\\
& d\in \mathcal D_{\text{HDC}}\tag{\ref{A}{d}}, \label{Ad}\\
& 0 \leq p_{i}\leq P_{\text{max}}, \quad\forall i\in \mathcal U \tag{\ref{A}{e}}, \label{Ae}\\
&0\le f_i\le f_i^{\max},\quad \forall i\in  \mathcal U \tag{\ref{A}{f}}, \label{Af}\\
& b_i\ge 0, \quad\forall i\in U \tag{\ref{A}{g}}, \label{Ag}\\
& t_i\ge 0, \quad\forall i\in \mathcal U \tag{\ref{A}{h}}, \label{Ah}\vspace{-0.48em}
\end{align}
where $\boldsymbol{t} = [t_1, \ldots, t_U]^T$,
$\quad \boldsymbol{b} = [b_1, \ldots, b_U]^T$,
$\boldsymbol{p} = [p_1, \ldots, p_U]^T$. $J_d$ represents the number of communication rounds required for a secure HDC-FL model of dimension $d$ to achieve the desired accuracy. Since $d$ is optimized in Problem~\eqref{A}, we use $J_d$ to explicitly denote the required training rounds as a function of $d$. Constraint \eqref{Aa} ensures that the combined local execution time and transmission time for each user does not exceed the maximum completion time $T$ for the whole FL algorithm. Since the FL algorithm is performed in parallel and the server updates the global AM only after receiving the uploaded local AMs from all users, this constraint synchronizes the parallel process. Constraint \eqref{Ab} indicates the restriction of the data transmission, while constraint \eqref{Ac} limits the sum of all allocated bandwidth to be smaller than the total amount of bandwidth.
The discrete set of admissible HDC dimensions is given by \eqref{Ad}. Furthermore, constraints \eqref{Ae} and \eqref{Af} limit the transmission power of each user and the maximum local computation frequency, respectively. Lastly, constraints \eqref{Ag} and \eqref{Ah} define the type of variable. 

\section{Algorithm Design}\label{AD}


To solve the optimization of dimension, bandwidth, power, and time, we propose an alternating optimization scheme that incorporates the HDC dimension. Initially, we characterize the relationship between the required convergence rounds $J_d$ and HDC dimension $d$. Subsequently, the optimization problem is decomposed into two iteratively subproblems: subproblem 1 focuses on optimizing the dimension $d$ given the fixed resource allocation $(\boldsymbol{t},\boldsymbol{b},\boldsymbol{p},\boldsymbol{f})$, and subproblem 2 focuses on optimizing resource allocation given $d$. By alternating between the two sets of variables, this scheme efficiently converges to the globally optimal solution $(\boldsymbol{d}^{*}, \boldsymbol{b}^{*}, \boldsymbol{p}^{*}, \boldsymbol{t}^{*}, \boldsymbol{f}^{*})$ that jointly minimizes the total energy.

\subsection{Empirical Modeling of Convergence Rounds}
In the total energy minimization problem \eqref{A}, the number of communication rounds, $J_d$, is a critical coupling variable that links the HDC dimension $d$ to both computation and communication energy consumption. While increasing the dimension $d$ enhances model robustness and accelerates convergence, this improvement is observed to be non-linear and subject to saturation at high dimensions. Since deriving an analytical relationship among dimension, privacy, and convergence remains challenging, we develop an empirical model to precisely estimate $J_d$ for any given $d$.

This empirical relationship is derived by observing how the convergence rounds $J_d$ vary with dimension $d$ under different privacy budgets $\epsilon$ and target accuracies ($\text{Acc}_{\text{target}}$). Based on the non-linear inverse relationship that asymptotically approaches a lower bound, as consistently revealed by multiple independent simulation runs and statistical averaging, we propose the following sigmoid-variant function to continuously approximate the convergence rounds $J_d$:
\vspace{-0.48em}
\begin{equation}
J_d(d) = \mu + \frac{\nu}{1 + e^{\beta(\log d - \alpha)}} 
 \label{J_d}\vspace{-0.48em}
\end{equation}
where $J_d$ denotes the total number of communication rounds required to achieve the target accuracy. The coefficients $\mu$, $\nu$, $\alpha$, and $\beta$ are obtained via non-linear least squares fitting against the simulated data and depend on the selected  $\text{Acc}_{\text{target}}$ and $\epsilon$. As illustrated in Fig. \ref{fig:fitting}, this model closely captures the actual relationship between $d$ and $J_d$.

We design this structure of this model to reflect the physical characteristics of HDC: the term involving the exponential function captures the rapid, non-linear performance gains observed when moving from low to moderate dimensions, while the constant parameter $\mu$ represents the minimum limiting number of rounds achievable in the high-dimensional regime.



\subsection{Iterative Algorithm}
The proposed iterative algorithm mainly contains two subproblems at each iteration. Subproblem 1 optimizes the dimension $d$ given the fixed resource allocation ($\boldsymbol{t}, \boldsymbol{b}, \boldsymbol{p}, \boldsymbol{f}$), and subproblem 2 optimizes ($\boldsymbol{t}, \boldsymbol{b}, \boldsymbol{p}, \boldsymbol{f}$) based on the obtained $d$ in the previous step.

\subsubsection{Subproblem 1: HDC Dimension Allocation for Fixed ($\boldsymbol{t}, \boldsymbol{b}, \boldsymbol{p}, \boldsymbol{f}$)} 
In this step, problem (\ref{A}) becomes:
\vspace{-0.48em}
\begin{align}
&\min_{_{d}}  \quad \sum_{i=1}^U (Y_i - R_i) d + R_i J_d d + Q_i J_d \label{d}\\
\text{s.t.}
&\frac{D_i d \left(\!C_i^{(\text{init})}\!+\!(J_d\!-\!1) C_i^{(\text{ret})\!}\right)}{f_i}\! + \!J_d t_i\!\leq \!T,\!\!\quad \forall i\in \mathcal U \tag{\ref{d}{a}}, \label{da}\\
& d\in \mathcal D_{\text{HDC}} \tag{\ref{d}{b}}. \label{db}
\vspace{-0.48em}
\end{align}
where the constants $Y_i = \gamma D_i C_i^{(init)} f_i^2$, $R_i = \gamma D_i C_i^{(ret)} f_i^2$, and $Q_i = t_i p_i$ are determined by the fixed resource allocation. And we denote $E_i(d) = (Y_i - R_i) d + R_i J_d d + Q_i J_d$.

To find the optimal $d$, we relax the discrete constraint \eqref{db} based on \eqref{J_d} and seek the optimal continuous dimension $\tilde{d}$ by setting the first-order derivative of the total energy $E(d)$ with respect to $d$ to zero. The derivative $\frac{\partial E(d)}{\partial d}$ is given by
\vspace{-0.48em}
\begin{equation}
    \frac{\partial E(d)}{\partial d} \!= \!\sum_{i=1}^{U} \!\left[ (Y_i \!- \!R_i) \!+ \!R_i \!\left( J_d'(d) d \!+ \!J_d \right)\! + \!Q_i J_d'(d) \right],\vspace{-0.48em}
\end{equation}
where
$J_d'(d) = \frac{\partial J_d(d)}{\partial d} = - \frac{\nu \beta e^{\beta(\log d - \alpha)}}{d \left(1 + e^{\beta(\log d - \alpha)}\right)^2}$. Setting the derivative to zero yields the following non-linear equation:
\vspace{-0.48em}
\begin{equation}
\sum_{i=1}^{U} \!\left[ (Y_i\! - \!R_i)\! + \!R_i \left( J_d'(d)d \!+\! J_d(d) \right)\! +\! Q_i J_d'(d) \right]\! = \!0.
\label{eq:nonlinear_equation_log_logistic} \vspace{-0.48em}
\end{equation}
We solve Eq. \eqref{eq:nonlinear_equation_log_logistic} via bisection to obtain candidate stationary points. To guarantee global optimality, we need to evaluate the total energy $E(d)$ at these points as well as at the boundary values determined by Eq. \eqref{da}. The optimal continuous dimension $\tilde{d}$ is the one that minimizes $E(d)$ among these candidates. Finally, $\tilde{d}$ is projected onto the closest dimension within the discrete feasible set $\mathcal D_{\text{HDC}}$.



\subsubsection{Resource Allocation for Fixed ${d}$} In the second stage, the $d$ and $J_d$ are fixed, the problem is to find the optimal parameters $(\boldsymbol t,\boldsymbol b,\boldsymbol p,\boldsymbol f)$ that minimize the total energy. We employ a secondary alternating update scheme, where we alternate between optimizing the uplink time $\boldsymbol{t}$ and optimizing the remaining resources $(\boldsymbol{f}, \boldsymbol{b}, \boldsymbol{p})$.

In the first step of this secondary iteration,
the CPU frequency vector \(\boldsymbol f\), bandwidth vector \(\boldsymbol b\) and the power vector \(\boldsymbol p\) are fixed.  Under this assumption, the original joint problem \eqref{A} decouples across users and reduces to a linear program in the uplink times \(\boldsymbol t\).  Specifically, we solve
\vspace{-0.48em}
\begin{align}
& \min_{\boldsymbol t}\quad \sum_{i=1}^U p_it_i, \label{B}\\
\text{s.t.}
&\frac{D_i d \left(\!C_i^{(\text{init})}\!+\!(J_d-1) C_i^{(\text{ret})}\!\right)}{f_i} \!+ \!J_d t_i\!\leq\! T,\!\!\quad\forall i\in \mathcal U \tag{\ref{B}{a}}, \label{Ba}\\
&t_ib_i\log_2\left(1+\frac{p_i g_i}{N_0b_i}\right)\geq N_id,\;\forall i\in  \mathcal U \tag{\ref{B}{b}}, \label{Bb}\\
  & t_i\ge 0, \;\forall i\in \mathcal U
  \tag{\ref{B}{c}}, \label{Bc}\vspace{-0.48em}
\end{align} 
where for any given $d$, $N_id$ can be treated as a constant. Because the objective in Eq. \eqref{B} is linear and monotonically increasing with respect to \(\boldsymbol t\) and each rate‐constraint
$b_i\log_2\!\Bigl(1 + \tfrac{p_i\,g_i}{N_0\,b_i}\Bigr)\,t_i \ge N_id$
imposes a lower bound:
\vspace{-0.48em}
\begin{equation}
  t_{i}^{\text{min}} = \frac{N_id}{b_i\log_2\bigl(1 + \frac{p_i g_i}{N_0 b_i}\bigr)},\quad\forall i\in  \mathcal U,
  \label{tmin}\vspace{-0.48em}
\end{equation}
To minimize the transmission energy in Eq. \eqref{B}, the optimal time can be derived using Theorem \ref{theorem1}.

\begin{theorem} The optimal solution $t^\star$ of problem \eqref{B} satisfies
\vspace{-0.48em}
\begin{equation}
  t_i^* = t_{i}^{\text{min}},\quad\forall i\in  \mathcal U,\vspace{-0.48em}
  \label{t_opt}
\end{equation}
where $t_i^{\min}$ is given by \eqref{tmin}, provided the feasibility condition
$t_i^{\min}\le \bar t_i$ holds for all $i$, with
$\bar t_i = \big(T-\frac{D_i d(C_i^{(\mathrm{init})}+(J_d-1)C_i^{(\mathrm{ret})})}{f_i}\big)\big/J_d$.
\label{theorem1}
\end{theorem}
\begin{proof}
The objective in problem \eqref{B} is $\sum_i p_i t_i$, which is strictly increasing in each $t_i$ since $p_i>0$.  
Constraint \eqref{Bb} implies the lower bound $t_i\ge t_i^{\min}$ in Eq. \eqref{tmin}. 
Thus any $t_i>t_i^{\min}$ can be decreased to $t_i^{\min}$ while preserving feasibility and strictly reducing the objective, contradicting optimality. Hence $t_i^\star=t_i^{\min}$ componentwise. The stated feasibility condition ensures that Eq. \eqref{Ba} is not violated at $t_i^{\min}$.
\end{proof}


In the second sub-step, given $(\boldsymbol{t},d)$, we solve the $(\boldsymbol{f},\boldsymbol{b},\boldsymbol{p})$. The optimization problem \eqref{A} can be rewritten as
\vspace{-0.48em}
\begin{align}
\min_{\boldsymbol {f}, \boldsymbol {b},\boldsymbol {p}}&
\sum_{i=1}^U \gamma Z_i d f_i^2 + (J_d-1)\gamma G_i d f_i^2 + J_d t_ip_i, \label{C}\\
\text{s.t.}&
\frac{Z_i d + (J_d -1)G_i d}{f_i}+ J_d t_i\le T,\quad
\forall i\in  \mathcal U \tag{\ref{C}{a}}, \label{Ca}\\
&t_ib_i\log_2\Bigl(1+\frac{g_i\,p_i}{N_0b_i}\Bigr)\ge N_id,\quad
\forall i\in  \mathcal U \tag{\ref{C}{b}}, \label{Cb}\\
&\sum_{i=1}^U b_i\le B,\quad\forall i\in  \mathcal U\tag{\ref{C}{c}}, \label{Cc}\\
&0\le p_i\le p_i^{\max},\quad\forall i\in  \mathcal U \tag{\ref{C}{d}}, \label{Cd}\\
&0\le f_i\le f_i^{\max},\quad \forall i\in  \mathcal U \tag{\ref{C}{e}}, \label{Ce}\vspace{-0.48em}
\end{align}
where $Z_i=D_iC_i^{(\text{init})}$ and $G_i=D_iC_i^{(\text{ret})}$ are constants. Since the objective function and constraints can be decoupled, problem \eqref{C} can be reformulated as two subproblems \eqref{D} and \eqref{E}.
\vspace{-0.48em}
\begin{align}
\min_{\boldsymbol {f}}&
\sum_{i=1}^U \left(Z_i  + (J_d -1)G_i\right)\gamma df_i^2 , \label{D}\\
\text{s.t.}&
\frac{Z_i d + (J_d-1)G_i d}{f_i}+ J_d t_i\le T,\quad
\forall i\in  \mathcal U \tag{\ref{D}{a}}, \label{Da}\\
&0\le f_i\le f_i^{\max},\quad \forall i\in  \mathcal U \tag{\ref{D}{b}}.\vspace{-0.48em}
\label{Db}
\end{align}

According to \eqref{D}, the objective function is strictly increasing in each $f_i$, thus the optimal solution is obtained by setting each $f_i$ to its lower bound, which is given as
follows:
\vspace{-0.48em}
\begin{align}
f_i^* &= \frac{Z_i d + (J_d-1)G_i d}{T-J_dt_i},
\quad \forall i\in\mathcal U.
\label{f*}\vspace{-0.48em}
\end{align}


The $(\boldsymbol{b}, \boldsymbol{p})$ optimization problem out of Eq. \eqref{C} can be formulated as:
\vspace{-0.48em}
\begin{align}
\min_{\boldsymbol {b},\boldsymbol {p}}&
\sum_{i=1}^U J_d t_ip_i,\label{E}\\
\text{s.t.}
&t_i\,b_i\log_2\Bigl(1+\frac{g_i\,p_i}{N_0\,b_i}\Bigr)\ge N_id,\quad
\forall i\in  \mathcal U \tag{\ref{E}{a}}, \label{Ea}\\
&\sum_{i=1}^U b_i\le B,\quad
\forall i\in  \mathcal U\tag{\ref{E}{b}}, \label{Eb}\\
&0\le p_i\le p_i^{\max},\quad\forall i\in  \mathcal U \tag{\ref{E}{c}}. \label{Ec}\vspace{-0.48em}
\end{align}

\begin{proposition} At optimality of Eq. \eqref{E}, the rate constraint \eqref{Ea} is tight for every user. Consequently,
\vspace{-0.48em}
\begin{equation}
p_i=\frac{N_0 b_i}{g_i}\left(2^{\frac{N_id}{t_i b_i}}-1\right), \qquad \forall i\in\mathcal U.
\label{p*}\vspace{-0.48em}
\end{equation}
\end{proposition}

Substituting \eqref{p*} into problem \eqref{E} yields a single-variable optimization problem in $\boldsymbol b$.
\vspace{-0.48em}
\begin{align}
\min _{\boldsymbol{b}} & \sum_{i=1}^U \frac{N_0 t_i b_i}{g_i}\left(2^{\frac{N_id}{b_i t_i}}-1\right), \label{F}\\
\text { s.t. } & \sum_{i=1}^U b_i \leq B,\quad \forall i \in \mathcal{U} \tag{\ref{F}{a}}, \label{Fa}\\
& b_i \geq b_i^{\min }, \quad \forall i \in \mathcal{U} \tag{\ref{F}{b}}. \label{Fb}\vspace{-0.48em}
\end{align}
The value $b_i^{\min}$ denotes the minimum required bandwidth to meet the transmission rate constraint when the maximum transmission power is used. Based on constraint \eqref{Ea}, $b_i^{\min}$ can be computed in closed form using the Lambert-W function, which is shown by
\vspace{-0.48em}
\begin{equation}
b_i^{\min}\!
=\!-\frac{(\ln2)N_id}
{t_i\!\mathrm{W}\!\Bigl(\!-\frac{(\ln2)N_0\,N_id}{g_ip_i^{\max}t_i}
e^{\!-\frac{(\ln2)N_0N_id}{g_ip_i^{\max}t_i}}\!\Bigr)\!\!+\!\frac{(\ln2)N_0N_id}{g_ip_i^{\max}}}.  \vspace{-0.48em}
\label{Lambert1}
\end{equation}

By computing the first and second order derivatives of problem \eqref{F}, named $F(b_i)$, we obtain:
\vspace{-0.48em}
\begin{equation}
\frac{\partial F(b_i)}{\partial b_i}\!=\!\frac{N_0 t_i}{g_i}\!\left(\!e^{\!\frac{(\ln 2) N_id}{t_i b_i}}\!\!-\!\frac{(\ln 2) N_id}{t_i b_i} e^{\!\frac{(\ln 2) N_id}{t_i b_i}}\!\!-\!1\right)\!< \!0,    \vspace{-0.48em}
\label{order1}
\end{equation}

\begin{equation}
\frac{\partial^2 F(b_i)}{\partial b_i^2}=\frac{N_0(\ln 2)^2 {N_i}^2 d^2}{g_i t_i b_i^3} e^{\frac{(\ln 2) N_id}{t_i b_i}} > 0.\vspace{-0.48em}
\label{order2}
\end{equation}
According to Eqs. \eqref{order1} and \eqref{order2}, it can be shown that the problem \eqref{F} is monotonically decreasing and convex with respect to $b_i>0$, which can be solved by using the Karush-Kuhn-Tucker (KKT) conditions. The Lagrangian function of \eqref{F} is
\vspace{-0.48em}
\begin{equation}
\mathcal{L}(\boldsymbol{b}, \lambda)=\sum_{i=1}^U \frac{N_0 t_i b_i}{g_i}\left(2^{\frac{N_id}{b_i t_i}}-1\right)+\lambda\left(\sum_{i=1}^U b_i-B\right),\vspace{-0.48em}
\label{KKT}
\end{equation}
where $\lambda$ is the Lagrange multiplier associated with constraint \eqref{Fa}. The first order derivative of $\mathcal{L}(b, \lambda)$ with respect to $b_i$ is
\vspace{-0.48em}
\begin{equation}
\frac{\partial \mathcal{L}(\boldsymbol{b}, \lambda)}{\partial b_i}\!=\!\frac{N_0 t_i}{g_i}\!\left(\!e^{\frac{(\ln 2) N_id}{t_i b_i}}\!\!\!-\!\frac{(\ln 2) N_id}{t_i b_i} e^{\frac{(\ln 2) N_id}{t_i b_i}}\!\!\!-\!1\!\right)\!\!+\!\lambda.   \vspace{-0.48em}
\label{KKT1}
\end{equation}

We define $b_i(\lambda)$ as the unique solution to $\frac{\partial \mathcal{L}(\boldsymbol{b}, \lambda)}{\partial b_i}=0$. Given constraint \eqref{Fb}, the optimal $b_i^*$ can be chosen as
\vspace{-0.48em}
\begin{equation}
b_i^* = \max\bigl\{b_i(\lambda),b_i^{\min}\bigr\},
\quad \forall i\in\mathcal U.\vspace{-0.48em}
\end{equation}
Moreover, the problem
\eqref{F} is a decreasing function of $b_i$, constraint \eqref{Fa} always holds with equality for the optimal solution, which shows that the optimal Lagrange multiplier should satisfies
\vspace{-0.48em}
\begin{equation}
\sum_{i=1}^U \max \left\{b_i(\lambda), b_i^{\min }\right\}=B.
\label{Lag}\vspace{-0.48em}
\end{equation}
Once $b_i^*$ is obtained, the corresponding $p_i^*$ can be computed by \eqref{p*}.

\begin{theorem}
Define $\Phi(\lambda)=\sum_{i\in\mathcal U}\max\{b_i(\lambda),b_i^{\min}\}$.
Then $\Phi(\lambda)$ is continuous and strictly decreasing in $\lambda$.
Therefore there exists a unique $\lambda^\star$ such that $\Phi(\lambda^\star)=B$, and
\vspace{-0.48em}
\[
b_i^\star=\max\{b_i(\lambda^\star), b_i^{\min}\},\quad
p_i^\star\ \text{from~\eqref{p*}},\vspace{-0.48em}
\]
solve problem~\eqref{F} globally.
\end{theorem}


Algorithm~\ref{alg:enum-alt} concludes that the proposed method solves problem \eqref{A} by alternately optimizing subproblem 1 and subproblem 2. Subproblem 2 is solved by alternately optimizing \eqref{B} and \eqref{C}. Since each subproblem achieves its optimal solution at every iteration, the objective function value of problem \eqref{A} decreases monotonically with each update.

\begin{algorithm}[t]
\caption{Iterative Alternating Optimization}
\label{alg:enum-alt}
\begin{algorithmic}[1]
\State Initialize a feasible solution $(d^{(0)},\boldsymbol{t}^{(0)},\boldsymbol{b}^{(0)},\boldsymbol{f}^{(0)},\boldsymbol{p}^{(0)})$ and set $k=0$.
\Repeat
    \State With given $(\boldsymbol{t}^{(k)},\boldsymbol{b}^{(k)},\boldsymbol{f}^{(k)},\boldsymbol{p}^{(k)})$, obtain the optimal $d^{(k+1)}$ (solve the stationarity condition and project to $\mathcal D_{\text{HDC}}$).
    \State Set $l=0$.
    \Repeat
        \State With given $(\boldsymbol{b}^{(l)},\boldsymbol{f}^{(l)},\boldsymbol{p}^{(l)})$, obtain the optimal $\boldsymbol{t}^{(l+1)}$.
        \State With given $\boldsymbol{t}^{(l+1)}$, obtain the optimal $(\boldsymbol{b}^{(l+1)},\boldsymbol{f}^{(l+1)},\boldsymbol{p}^{(l+1)})$.
        \State Set $l=l+1$.
    \Until the inner objective value converges or $l \ge L_{\mathrm{in}}^{\max}$.
    \State Set $(\boldsymbol{t}^{(k+1)},\boldsymbol{b}^{(k+1)},\boldsymbol{f}^{(k+1)},\boldsymbol{p}^{(k+1)}) \leftarrow (\boldsymbol{t}^{(l)},\boldsymbol{b}^{(l)},\boldsymbol{f}^{(l)},\boldsymbol{p}^{(l)})$.
    \State Set $k=k+1$.
\Until the outer objective value converges or $k \ge L_{\mathrm{out}}^{\max}$.
\State \textbf{Output:} $(d^\star,\boldsymbol{t}^\star,\boldsymbol{b}^\star,\boldsymbol{f}^\star,\boldsymbol{p}^\star)$.
\end{algorithmic}
\end{algorithm}

\subsection{Complexity Analysis}
Let $U$ be the number of users. The algorithm runs $L_{\mathrm{out}}$ outer alternations between updating $d$ and updating $(\boldsymbol{t},\boldsymbol{b},\boldsymbol{f},\boldsymbol{p})$. The $d$-update is a one-dimensional root-finding step with $\mathcal{O}(U)$ cost per function (or derivative) evaluation, and thus costs $\mathcal{O}\big(U L_d\big)$ with $L_d=\mathcal{O}(\log(1/\epsilon_d))$. For a fixed $d$, the resource update performs $L_{\mathrm{in}}$ inner alternations, where the bandwidth KKT procedure dominates: bisection finds $\lambda^\star$ in $L_\lambda=\mathcal{O}(\log(1/\epsilon_\lambda))$ iterations, and each $\lambda$-iteration computes ${b_i(\lambda)}{i=1}^{U}$ by solving $U$ monotone scalar equations in $L_b=\mathcal{O}(\log(1/\epsilon_b))$ steps \cite{1664999}. Hence, each inner iteration costs $\mathcal{O}(U L_b L\lambda)=\mathcal{O}\big(U\log(1/\epsilon_b)\log(1/\epsilon_\lambda)\big)$, and the overall complexity is $\mathcal{O}\Big(L_{\mathrm{out}}U\big(\log(1/\epsilon_d)+L_{\mathrm{in}}\log(1/\epsilon_b)\log(1/\epsilon_\lambda)\big)\Big)$.

\section{Algorithm to Finding a Feasible Solution of Problem \eqref{A}}\label{fea}
Since the proposed alternating energy-minimization algorithm requires a feasible starting point, this section presents an initialization procedure under the bandwidth budget $B$ and time limit $T$. 



\subsection{Dimension Initialization and Feasibility Certificate}
We choose an initial $d^{(0)}\in\mathcal D_{\text{HDC}}$ (e.g., a moderate value, or the smallest value in $\mathcal D_{\text{HDC}}$). Given $d^{(0)}$, the corresponding convergence round number $ J_d(d^{(0)})$ is obtained from Eq. \eqref{J_d}.

To certify feasibility under the CPU frequency cap, we define the maximum admissible per-round communication time $t_i^{\max}(d)$ as:
\vspace{-0.48em}
\begin{equation}
t_i^{\max}(d)
=
\frac{
T-\frac{Z_i d+(J_d-1)G_i d}{f_i^{\max}}
}{
J_d
},\quad \forall i\in\mathcal{U}.\vspace{-0.48em}
\label{eq:tmax1}
\end{equation}
If $t_i^{\max}(d)\le 0$ for some $i$, then the instance is infeasible for this $d$. 

Next, fix $p_i=p_i^{\max}$ and let $b_i^{\min}(t(d))$ denote the minimum bandwidth required to satisfy the rate constraint within a transmission time $t$ for dimension $d$, which can be obtained in closed form via the Lambert W function \eqref{Lambert1}. We then have the following proposition.

\begin{proposition}
\label{prop:feas_cert}
Given $(B,T)$ and per-user limits $\{f_i^{\max},p_i^{\max}\}$, the constraint set of problem \eqref{A} is nonempty for a fixed $d$ if and only if
\begin{enumerate}
    \item $t_i^{\max}(d) > 0$, $\forall i\in\mathcal{U}$; and
    \item $\sum_{i=1}^{U} b_i^{\min}(t_i^{\max}(d)) \le B$.
\end{enumerate}
\end{proposition}

\begin{proof}
If \eqref{eq:tmax1} yields $t_i^{\max}(d)\le 0$ for some $i$, then even with $f_i=f_i^{\max}$ the local computation time already exceeds the available budget, hence infeasible. Otherwise, $t_i\le t_i^{\max}(d)$ must hold for feasibility (since $f_i\le f_i^{\max}$). For any such $t_i$, the rate constraint requires $b_i\ge b_i^{\min}(t_i(d))$, in particular, $b_i\ge b_i^{\min}(t_i^{\max}(d))$. Therefore, feasibility implies $\sum_i b_i^{\min}(t_i^{\max}(d))\le B$.
\end{proof}

If either condition fails, the $d$ is infeasible. If no candidate in $\mathcal D_{\text{HDC}}$ passes the certificate, the instance is infeasible for the given $(B, T)$.


\subsection{Initialization of $(\boldsymbol{b},\boldsymbol{p},\boldsymbol{t},\boldsymbol{f})$ for a Given $d$}
\label{subsec:init_bptf_given_d}
Given a dimension $d$ that passes the feasibility certificate in Proposition~\ref{prop:feas_cert}, we construct an initial resource point $(\boldsymbol{b}^{(0)},\boldsymbol{p}^{(0)},\boldsymbol{t}^{(0)},\boldsymbol{f}^{(0)})$ as follows. 

We allocate bandwidth by minimizing the total transmission time subject to the bandwidth budget and per-user feasibility lower bounds. In particular, with $p_i^{(0)}=p_i^{\max}$, we solve
\vspace{-0.48em}
\begin{align}
&\min_{\boldsymbol b,\boldsymbol t}\quad
\sum_{i=1}^U t_i \label{Mint}\\
\text{s.t.}\quad
&t_ib_i\log_2\left(1+\frac{c_i}{b_i}\right)\geq N_id, \forall i\in  \mathcal U \tag{\ref{Mint}{a}}, \label{Minta}\\
& \sum_{i=1}^U b_i \leq B, \tag{\ref{Mint}{b}} \label{Mintb}\\
& b_i \ge b_i^{\min}\big(t_i^{\max}(d)\big),\ t_i\ge 0,\ \forall i\in\mathcal{U}
\tag{\ref{Mint}{c}},\label{Mintc}\vspace{-0.48em}
\end{align}
where $c_i= \frac{p_i^{\max}g_i}{N_0}$. Since the objective $\sum_i t_i$ is strictly increasing in each $t_i$, the rate constraint \eqref{Minta} is tight at optimum. Therefore, for any given $b_i>0$, the optimal transmission time is
$t_i(b_i)=\frac{N_id}{b_i \log_2\!\left(1+\frac{c_i}{b_i}\right)}$. Therefore, problem \eqref{Mint} can be rewritten as:
\vspace{-0.48em}
\begin{align}
  &\min_{\boldsymbol b}\quad
   \sum_{i=1}^U \frac{N_id}{b_i \log _2\left(1+c_i / b_i\right)} \label{Mint1}
  \\
  \text { s.t. }&\eqref{Mintb}-\eqref{Mintc}
\tag{\ref{Mint1}{a}}.
\label{Mint1a}\vspace{-0.48em}
\end{align}

\begin{lemma}
\label{lem:b_init_kkt}
The function $\phi_i(b_i)= \frac{N_id}{b_i\log_2\left(1+\frac{c_i}{b_i}\right)}$ is convex and strictly decreasing in $b_i>0$. Hence, problem~\eqref{Mint1} is convex and admits a unique optimal solution $\boldsymbol{b}^{(0)}$ characterized by 
\vspace{-0.48em}
\begin{equation}
\phi_i'(b_i^{(0)})\!+\!\lambda^\star\!=\!0,\,\, b_i^{(0)}\! >\! b_i^{\min}\big(t_i^{\max}(d)\big),\, \,\sum_{i=1}^U b_i^{(0)}\! =\! B,
\label{eq:kkt_primal_lb}\vspace{-0.48em}
\end{equation}
where \(\lambda^\star\ge 0\) is the Lagrange multiplier.
\end{lemma}

\begin{proof}
Convexity and monotonicity of $\phi_i(\cdot)$ follow from differentiation. Since the constraints in \eqref{Mint1} are convex, \eqref{Mint1} is a convex program and the KKT conditions are necessary and sufficient. Uniqueness follows from strict convexity of $\phi_i(\cdot)$ on $b_i>0$. Finally, because $\phi_i(\cdot)$ is strictly decreasing, the sum-bandwidth constraint is tight at optimum.
\end{proof}

After obtaining $\boldsymbol{b}^{(0)}$, we then compute the transmission time by the tight rate constraint:
\vspace{-0.48em}
\begin{equation}
t_i^{(0)}=\frac{N_id}{b_i^{(0)} \log_2\!\left(1+\frac{c_i}{b_i^{(0)}}\right)},\quad \forall i\in\mathcal{U},
\vspace{-0.48em}\label{eq:t_init}
\end{equation}
and set the CPU frequency by enforcing the time constraint with equality:
\vspace{-0.48em}
\begin{equation}
f_i^{(0)}=\frac{Z_i d+(J_d-1)G_i d}{T-J_d\,t_i^{(0)}},\quad \forall i\in\mathcal{U}.
\label{eq:f_init}\vspace{-0.48em}
\end{equation}

\begin{proposition}
\label{prop:init_bptf_feasible}
If $d^{(0)}$ satisfies Proposition~\ref{prop:feas_cert} and $\boldsymbol{b}^{(0)}$ is obtained by solving \eqref{Mint1} with $p_i^{(0)}=p_i^{\max}$, then the constructed $(\boldsymbol{b}^{(0)},\boldsymbol{p}^{(0)},\boldsymbol{t}^{(0)},\boldsymbol{f}^{(0)})$ is feasible for Problem \eqref{A} under dimension $d^{(0)}$.
\end{proposition}

\begin{proof}
Since $p_i^{(0)}=p_i^{\max}$ by initialization and Eq. \eqref{Mintb} ensures $\sum_i b_i^{(0)}\le B$, the power and sum-bandwidth constraints are satisfied.
Eq.~\eqref{eq:t_init} enforces the rate constraint with equality. Moreover, from Eq. \eqref{Mintc} and the definition of $b_i^{\min}(\cdot)$, the induced transmission time satisfies $t_i^{(0)}\le t_i^{\max}(d)$ for all $i$. Therefore,
$
T-J_d t_i^{(0)} \ge T-J_d t_i^{\max}(d)
= \frac{Z_i d+(J_d-1)G_i d}{f_i^{\max}},
$
and substituting into \eqref{eq:f_init} yields $f_i^{(0)}\le f_i^{\max}$. Hence, all constraints in problem~\eqref{A} are satisfied.
\end{proof}

\subsection{Complexity
Analysis}\label{Complexity Analysis}
For a given candidate $d$, computing ${t_i^{\max}(d)}$ and ${b_i^{\min}(t_i^{\max}(d))}$ for all users costs $\mathcal{O}(U)$. The dominant cost is solving the bandwidth initialization problem~\eqref{Mint1}. By Lemma~\ref{lem:b_init_kkt}, it admits a KKT characterization with a scalar multiplier $\lambda^\star$, found by bisection in $L_{\lambda_1}=\mathcal{O}(\log(1/\epsilon_{\lambda_1}))$ iterations \cite{1664999}. In each iteration, ${b_i(\lambda)}$ are obtained by solving $U$ monotone scalar equations to accuracy $\epsilon_{b_1}$ in $L_{b_1}=\mathcal{O}(\log(1/\epsilon_{b_1}))$ steps, yielding complexity $\mathcal{O}(U L_{b_1} L_{\lambda_1})$. After $\boldsymbol{b}^{(0)}$ is obtained, updating $\boldsymbol{t}^{(0)}$ and $\boldsymbol{f}^{(0)}$ is closed-form and costs $\mathcal{O}(U)$. Therefore, the initialization complexity for a given $d$ is $\mathcal{O}(U L_{b_1} L_{\lambda_1})$, and checking $L_{d_1}$ candidate dimensions in $\mathcal{D}$ leads to $\mathcal{O}(L_{d_1}U L_{b_1} L_{\lambda_1})$.

\section{Simulation Results and Analysis}\label{simulation}

\subsection{Simulation Setup}
We consider a secure FL-HDC system in which $K=50$ users are independently and uniformly distributed over a circular area of radius $500$\,m, with the BS located at the center. The wireless link adopts a standard large-scale path-loss model $L\left(d_i\right)=128.1+37.6 \log _{10}\left(d_i\right)$ \cite{9264742}. Uplink access is orthogonal within a total bandwidth $B$, and each user's transmit power is constrained by $P_{\max}$. Local computation energy is modeled using a standard frequency–energy coupling surrogate that captures the fact that lowering the CPU frequency saves energy at the expense of longer execution time \cite{7572018}. The per-dimension cycle costs in the HDC pipeline: encoding, first-round class aggregation, similarity evaluation, and error-driven update are denoted $C_{\mathrm{enc}}$, $C_{\mathrm{agg}}$, $C_{\mathrm{sim}}$, and $C_{\mathrm{up}}$, respectively. All system and computation parameters are summarized in Table~\ref{tab:sim_params} and are used as defaults unless otherwise specified.

\begin{table}[!t]
\vspace{-0.68em}
\centering
\caption{Simulation parameters}
\vspace{-0.68em}
\label{tab:sim_params}
\renewcommand{\arraystretch}{1.1}
\begin{tabular}{l l}
\hline
\textbf{Parameters} & \textbf{Value} \\
\hline
$K$: number of users & $50$ \\
$C_{\mathrm{enc}}$: encoding cycles per dimension & $28*28*2$  \\
$C_{\mathrm{agg}}$:  1st-round aggregation cycles per dimension& $28*28*2$ \\
$C_{\mathrm{sim}}$: similarity evaluation cycles per dimension & $10*10$  \\
$C_{\mathrm{up}}$: error-driven update cycles per dimension& $8$  \\
$e_i$: inference error ratio for user $i$ & $0.4$ \\
$d$ : HV dimension range & 3000:1000:10000\\
$\eta$: learning rate of HDC model & $ 1$\\
$T$: total time budget& 30 s \\
$f_i^{\max}$: maximum CPU frequency of user $i$  & $2.3$ GHz \\
$\gamma$: switched-capacitance coefficient & $1\times10^{-28}$ \\
$B$: total bandwidth & $10$ MHz \\
$P_{\max}$: maximum transmit power & $1$ mW \\
$n_0$: noise power spectral density & $-174$ dBm/Hz \\
\hline
\end{tabular}
\vspace{-1.68em}
\end{table}

Our simulation comprises three parts. 
\begin{itemize}
    \item We evaluate the performance of the proposed FL-HDC-DP framework, analyze its convergence and the impact of the HV dimension. Moreover, comparing it with an FL-NN-DP baseline under different privacy budgets and data distributions.
    \item We fit the relationship between $d$ and the number of convergence rounds $J_d$ using the sigmoid-variant model to validate its accuracy and suitability for HDC convergence efficiency under target accuracy and privacy budget.  
    \item We assess the advantages of the proposed joint dimension, bandwidth, transmit power, transmission time, and CPU frequency optimization method in reducing total energy within the FL-HDC-DP framework while satisfying accuracy and privacy requirements. 
\end{itemize}

\subsection{Performance of FL-HDC-DP}
For the performance study, both IID and non-IID MNIST client data are considered. Under IID, the training set is uniformly split across $U=50$ users (1200 samples each). Under non-IID \cite{mcmahan2017communication}, the 60{,}000 training images are label-sorted and partitioned into $3U$ shards (400 images per shard), with each client randomly assigned three shards without replacement. Client uploads are protected by a Gaussian DP mechanism with $\varepsilon\in\{10,20\}$ and $\delta=10^{-5}$, where the noise is pre-calibrated and not optimized. 

Fig.~\ref{fig:acc} illustrates the accuracy versus epochs under IID data for different $d$, highlighting the impact of both dimensionality and DP noise on convergence speed and final accuracy. Fig.~\ref{fig:acc}\subref{fig:acc_FL-HDC} shows the results of FL-HDC without DP noise: as $d$ increases from 3000 to 10000, the accuracy improves from about 92\% to approximately 95\%, and the number of epochs required to reach the 92\% threshold decreases from around 20 to about 10. Fig.~\ref{fig:acc}\subref{fig:acc_FL-HDC-DP} presents the results of FL-HDC-DP with DP noise ($\varepsilon=20$): the overall accuracy drops and becomes more oscillatory, but larger $d$ remains consistently beneficial, achieving faster convergence and higher accuracy at the same epoch (e.g., reaching 88\% requires about 55\% fewer epochs at $d=10000$ than at $d=3000$). Overall, higher-dimensional HVs enhance representation and reduce inter-class interference, but they also increase computation and communication overhead, requiring a practical trade-off between accuracy and cost while ensuring privacy.



\begin{figure}[!t]
  \vspace{-1.28em}
  \centering
  \subfloat[Without DP noise]{%
    \includegraphics[width=0.485\linewidth]{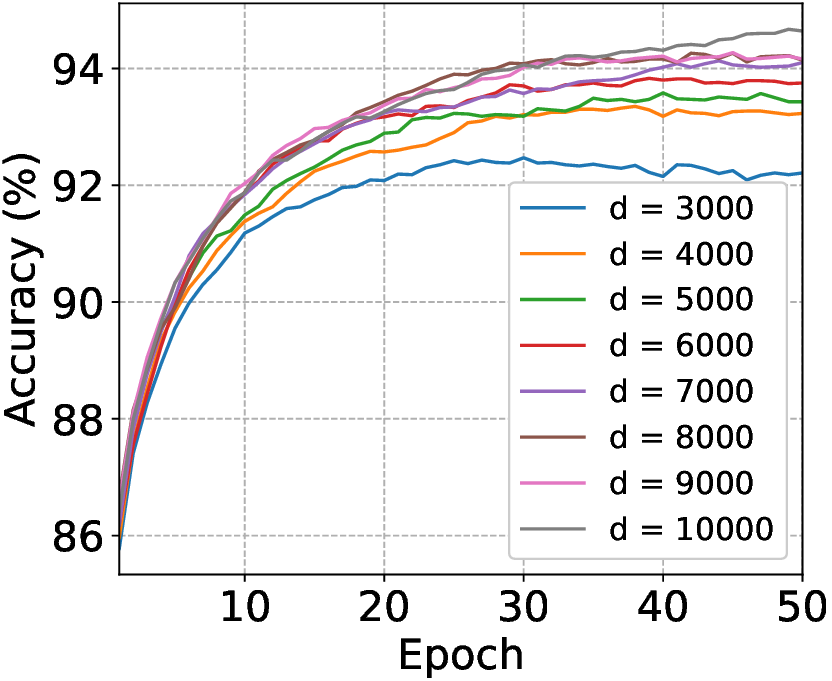}%
    \label{fig:acc_FL-HDC}
  }
  \hfill
  \subfloat[With DP noise ($\epsilon = 20$)]{%
    \includegraphics[width=0.485\linewidth]{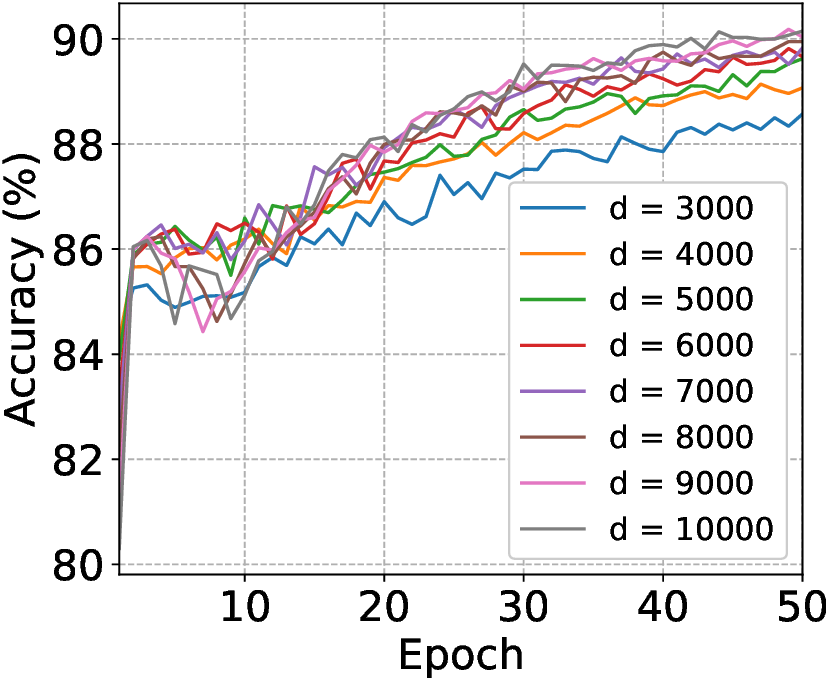}%
    \label{fig:acc_FL-HDC-DP}
  }
  \vspace{-0.48em}
  \caption{Accuracy vs. Epochs of FL-HDC under IID data with different hypervector dimensions: (a) without DP noise and (b) with DP noise ($\epsilon = 20$).}
  \vspace{-1.28em}
  \label{fig:acc}
\end{figure}



\begin{figure}[!t]
  \vspace{-0.6em}
  \centering
  \subfloat[IID and non-IID data, ($\epsilon = 20$)]{%
    \includegraphics[width=0.485\linewidth]{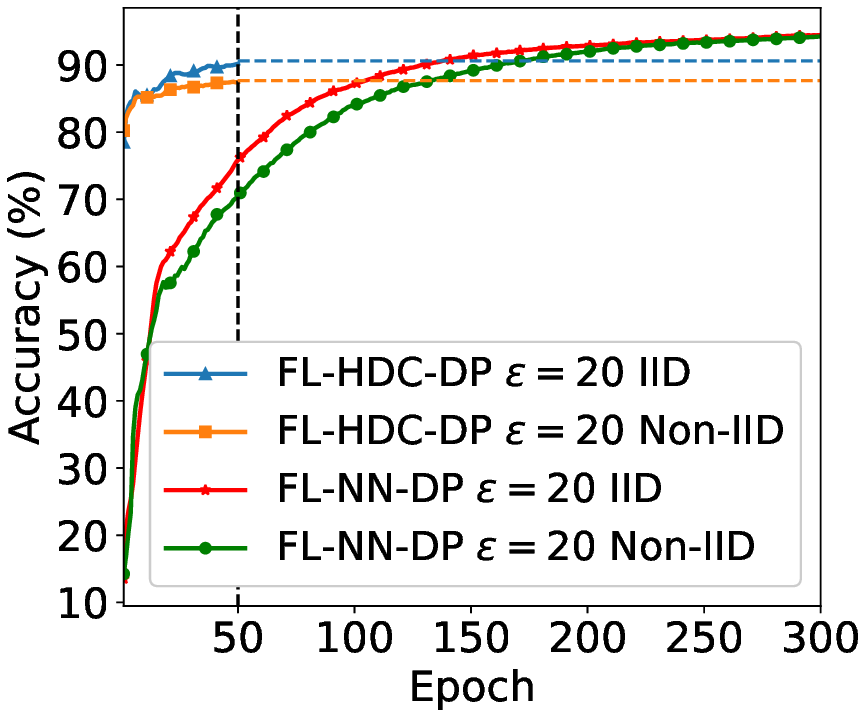}%
    \label{fig:acc1}
  }
  \hfill
  \subfloat[non-IID data, ($\epsilon =10, 20$)]{%
    \includegraphics[width=0.485\linewidth]{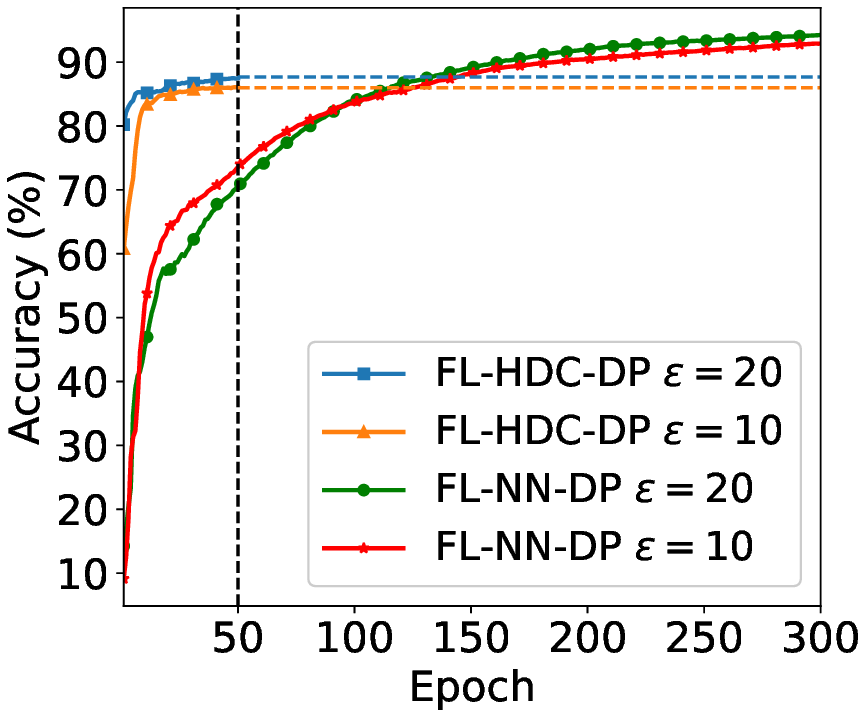}%
    \label{fig:acc2}
  }
  \vspace{-0.48em}
  \caption{Accuracy vs. epochs of FL-HDC-DP and FL-NN-DP under different data distributions and privacy budgets: (a) IID and non-IID data with $\epsilon = 20$ and (b) non-IID data with $\epsilon = 10, 20$.}
  \vspace{-1.9em}
  \label{fig:acc22}
\end{figure}

Fig.~\ref{fig:acc22} compares the accuracy versus epochs of FL-HDC-DP and FL-NN-DP under privacy constraints across different data distributions and privacy budgets. 
Fig.~\ref{fig:acc22}\subref{fig:acc1} shows that, under the same privacy budget ($\epsilon=20$), FL-HDC-DP converges faster than FL-NN-DP for both IID and non-IID partitions: under IID, it exceeds 80\% accuracy in the first round and reaches about 90\% within 40 rounds, whereas FL-NN-DP typically needs around 140 rounds to reach 90\% accuracy, i.e., about $3.5\times$ more rounds. 
Under non-IID data, both methods converge more slowly and exhibit lower steady-state accuracy, yet FL-HDC-DP remains consistently higher throughout training. 
Fig.~\ref{fig:acc22}\subref{fig:acc2} further evaluates non-IID performance under two privacy budgets and confirms the privacy--accuracy trade-off: at round 50, FL-HDC-DP achieves about 88\% for $\epsilon=20$ and 86\% for $\epsilon=10$, while FL-NN-DP attains about 72\% and 69\%, respectively. 
Moreover, FL-NN-DP typically requires 120--150 rounds to catch up and reaches about 92\%--94\% by 300 rounds. 
Note that the dimension of FL-HDC-DP in Fig.~\ref{fig:acc22}\subref{fig:acc1} and Fig.~\ref{fig:acc22}\subref{fig:acc2} is $d=10000$.

\subsection{Convergence Curve Fitting}
Fig.~\ref{fig:fitting} illustrates the relationship between the dimension $d$ and the number of rounds to converge $J_d$ under different target accuracies, and validates the fitting performance of the proposed sigmoid-based model. Fig.~\eqref{fig:conv_no_dp} corresponds to the setting without DP noise, while Fig. \eqref{fig:conv_dp} shows the case with DP noise ($\epsilon = 20$). We observe that, for the same dimension and accuracy threshold, introducing DP noise increases the number of rounds required to reach the target accuracy, i.e., convergence is noticeably delayed. However, in both the DP and non-DP settings, $J_d$ exhibits a similar trend with respect to $d$: it first decreases rapidly and then gradually saturates as the dimension increases. By fitting all curves using Eq. \eqref{J_d}, the proposed sigmoid-shaped function has good generality: for any given privacy budget and target accuracy, simply adjusting the model parameters allows us to accurately characterize the relationship between the $d$ and $J_d$.

\subsection{Total Energy Minimization}

For the optimization study, clients use an IID partition. The privacy budget is fixed to $(\varepsilon,\delta)=(25,10^{-5})$. We minimize the end-to-end training energy by jointly optimizing the resource and model variables, hypervector dimension $d$, bandwidth allocation, transmit power, and CPU frequency, subject to (i) a target test accuracy of $88\%$ and (ii) the fixed DP requirement $(\varepsilon=25,\delta=10^{-5})$. The DP noise variance is calibrated in advance to satisfy the budget and is not treated as a decision variable. 
Based on the convergence results under this setting, we fit the dimension–convergence relationship using Eq. \eqref{J_d} obtaining the parameters
$\mu = 15.25$, $\nu = 99.99$, $\alpha = 7.80$, and $\beta = 5.69$, which are used throughout the subsequent optimization.



\begin{figure}[!t]
\vspace{-1.68em}
  \centering
  \subfloat[Without DP noise]{%
    \includegraphics[width=0.485\linewidth]{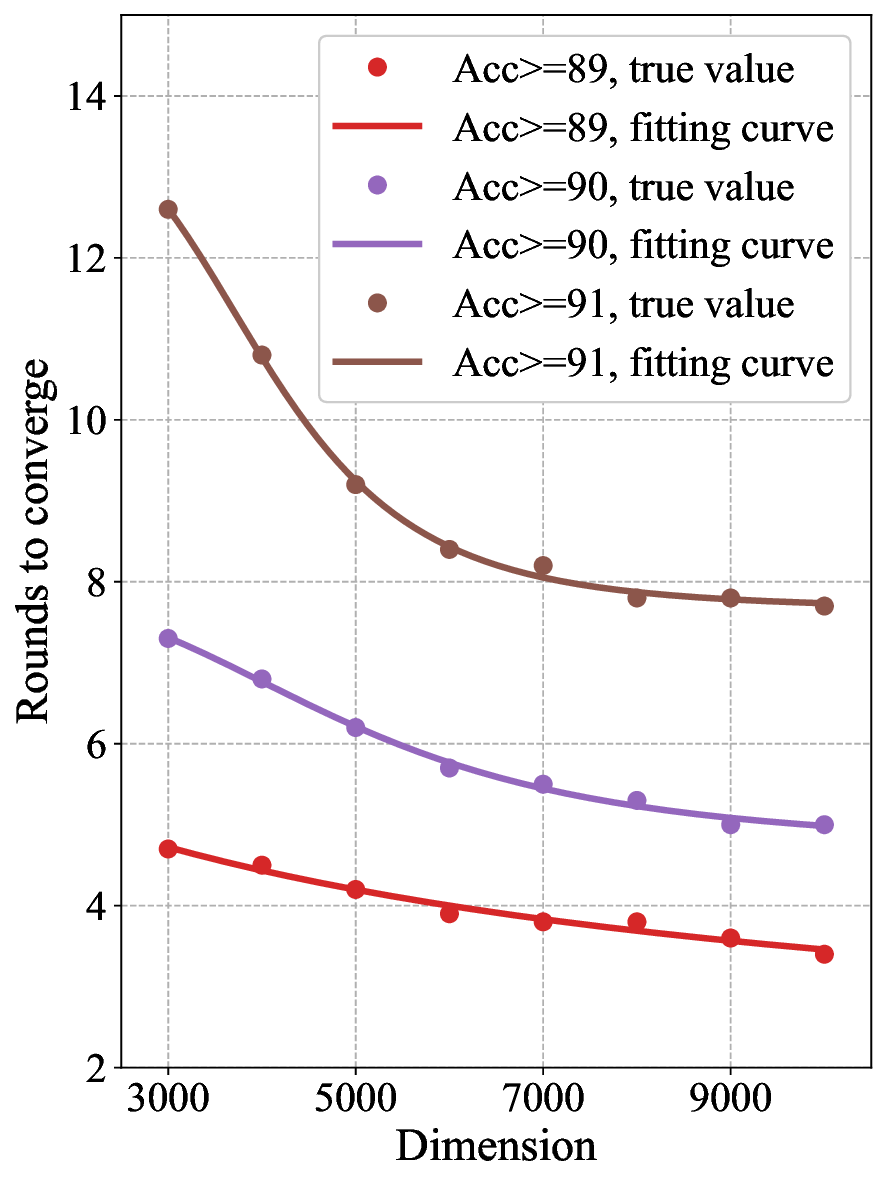}%
    \label{fig:conv_no_dp}
  }
  \hfill
  \subfloat[With DP noise ($\epsilon = 20$)]{%
    \includegraphics[width=0.485\linewidth]{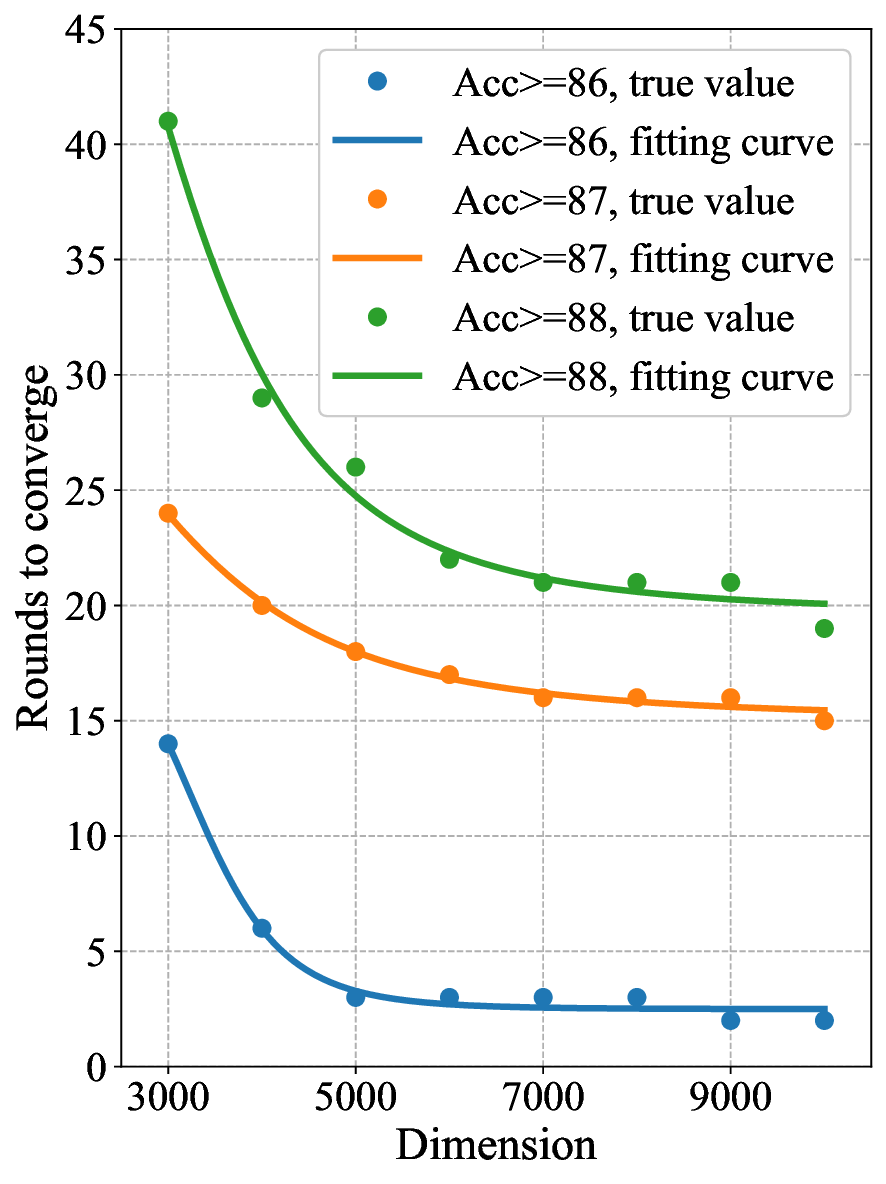}%
    \label{fig:conv_dp}
  }
  \vspace{-0.46em}
  \caption{Fitting performance of the proposed sigmoid-variant model for the relationship between the hypervector dimension $d$ and the number of rounds to converge $J_d$ under different target accuracies: (a) without DP noise and (b) with DP noise ($\epsilon = 20$).}
  \label{fig:fitting}
\end{figure}

Fig.~\ref{fig:dimension} demonstrates the total energy versus the dimension $d$ for several total bandwidths, where bandwidth, transmit power, and CPU frequency are optimized to meet a fixed target accuracy. It is shown that total energy is non-monotonic in $d$ and reaches a minimum at $d=4000$. According to Eq. \eqref{J_d}, increasing $d$ from 3000 to 4000 reduces the required global rounds by about 20; this reduction outweighs the higher per-round computation and transmission cost, leading to lower total energy. Increasing $d$ further to 5000 reduces rounds by only about 5 more, which is insufficient to offset the increased per-round energy, so the total energy rises. At low bandwidths (0.5-1~MHz), transmission remains the bottleneck, yielding higher curves and a more pronounced turning point, whereas at higher bandwidths (5/10/20~MHz) communication constraints ease and the curves shift downward and flatten, with diminishing gains from further bandwidth increases.

\begin{figure}[!t]
    \vspace{-1.68em}
    \centering
    \includegraphics[width=0.4\textwidth]{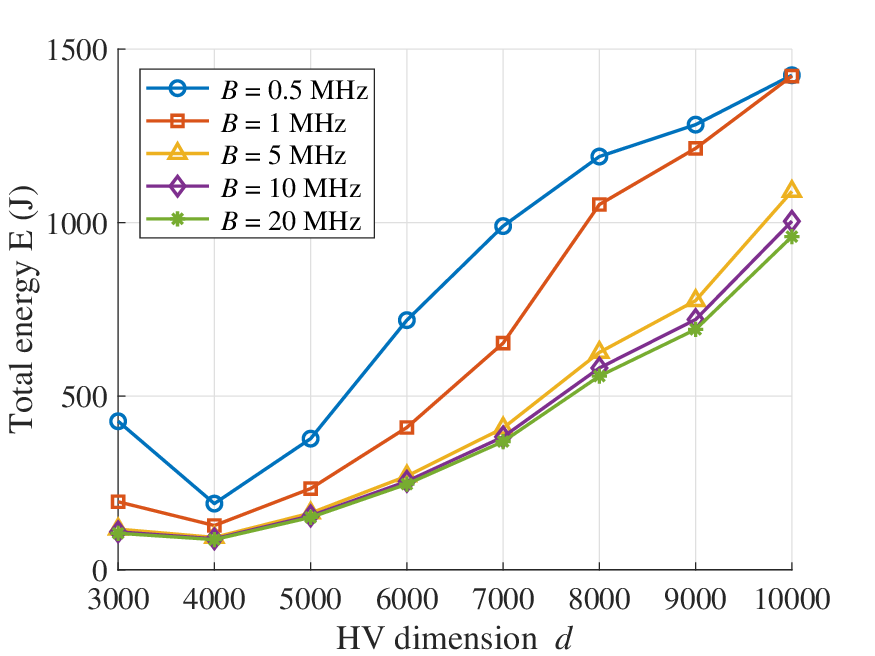}
    \caption{Total energy vs. hypervector dimension under different total bandwidth values.}
    \label{fig:dimension}
\end{figure}

\begin{figure}[!t]
    \centering
    \includegraphics[width=0.4\textwidth]{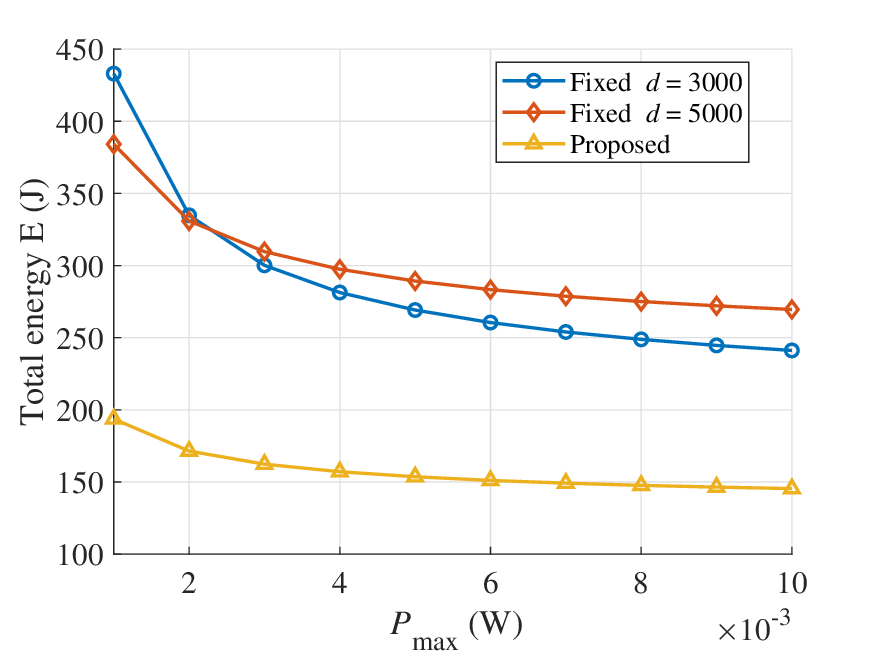}
    \vspace{-0.88em}
    \caption{Total energy vs. maximum transmit power.}
    \label{fig:power}
    \vspace{-1.68em}
\end{figure}

Fig.~\ref{fig:power} compares the total energy versus the maximum transmit power $P_{\max}$ for the proposed scheme and two fixed-dimension baselines ($d=3000$ and $d=5000$). The total energy decreases with $P_{\max}$, dropping sharply at small $P_{\max}$ and then saturating. This is because higher power increases the uplink rate and shortens transmission, leaving more time for local computation; once the rate constraint is sufficiently relaxed, computation dominates and further power yields diminishing returns. The proposed scheme achieves the lowest energy, about 40\% lower than the fixed $d=3000$ baseline and 46\% lower than the fixed $d=5000$ baseline, by selecting optimal $d=4000$ in this setting. The results also indicate that total energy is not monotonic in $d$, consistent with Fig.~\ref{fig:dimension}.

\begin{figure}[!t]
    \vspace{-1.68em}
    \centering    \includegraphics[width=0.4\textwidth]{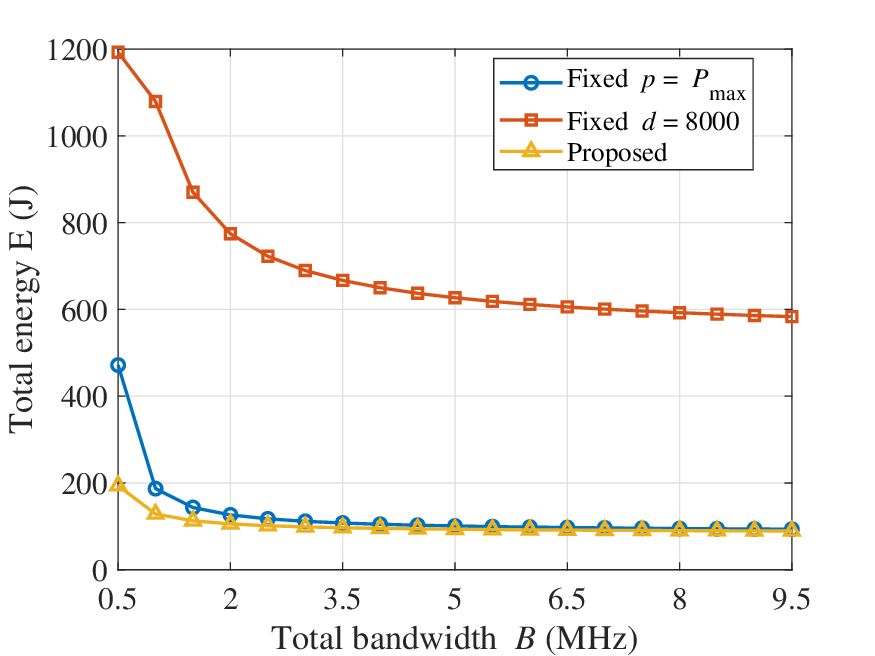}
    \vspace{-0.88em}
    \caption{Total energy vs. total bandwidth.}
    \label{fig:Bandwidth}
    \vspace{-1.28em}
\end{figure}

\begin{figure}[!t]
    \centering
\includegraphics[width=0.45\textwidth]{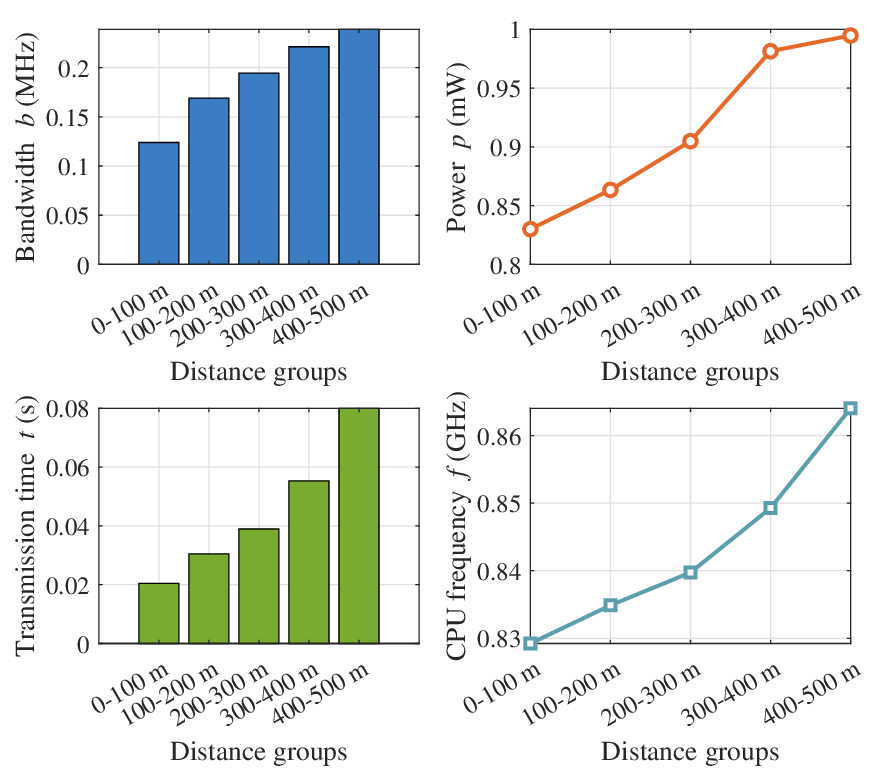}
    \vspace{-1.1em}
    \caption{Optimized average values of bandwidth, transmit power, transmission time, and CPU frequency across different distance groups.}
    \label{fig:allparameter}
    \vspace{-1.5em}
\end{figure}

Fig.~\ref{fig:Bandwidth} shows the total energy versus the total bandwidth $B$. For all schemes, energy decreases with $B$, dropping sharply over 0.5-2~MHz and then leveling off, since larger bandwidths improve the uplink rate and shorten the transmission time; once the rate constraint becomes non-binding, local computation dominates and further bandwidth offers limited benefit. The proposed scheme consistently achieves the lowest energy, saving up to 83.3\% compared with the fixed-$d$ baseline and about 57\% compared with the fixed-$p$ baseline. This gain comes from jointly optimizing $d$ and $p$, whereas the blue baseline fixes $p=P_{\max}$ and the red baseline fixes $d=8000$. The comparison also reveals that optimizing $d$ affects both computation and transmission energy, while optimizing bandwidth or power mainly reduces the transmission component, leading to larger overall savings from dimension optimization.

\subsection{Optimization Strategy of Proposed Method}
Fig. \ref{fig:allparameter} summarizes how the optimizer allocates bandwidth, transmit power, transmission time, and CPU frequency to uniformly distributed users as the link quality degrades with distance. A consistent pattern emerges. Bandwidth increases with distance, far users receive more spectrum to compensate for lower SNR and reduced spectral efficiency, which is an energy-efficient adjustment strategy compared with pushing power alone. Transmit power also increases with distance, with edge users approaching the device limits, indicating that when rate and latency targets become challenging, the system simultaneously employs power amplification and additional bandwidth strategies. Even so, transmission time for distant users still extends because transmitting the same payload over weak channels requires more time. 
This extended transmission time compresses computation time, prompting the optimizer to increase CPU frequency to ensure local processing completes within time constraints. Overall, distant users are more resource-hungry, consuming higher energy to meet the performance standards maintained by near users.

\section{Conclusion}\label{con}
In this paper, we have proposed a secure FL framework based on HDC and DP for resource-constrained edge devices in wireless networks. We have formulated a joint optimization of the HDC dimension, transmission time, transmit power, bandwidth, and CPU frequency to minimize total energy consumption, which is the first study to consider the impact of HDC dimension on energy optimization. To address this issue, we have proposed a two-level alternating optimization scheme. First, we have established the relationship between HDC dimensions and the number of iterations required to achieve target accuracy. Subsequently, an outer alternating optimization is applied to obtain the dimension, and an inner alternating optimization is applied to obtain resource allocation at a fixed dimension, where closed-form update formulas for computational and communication resources are derived for each iteration. We have further proposed a method for finding initial solutions by minimizing transmission time. Numerical results have demonstrated that FL-HDC-DP converges faster than traditional NN-based approaches while maintaining high accuracy. Moreover, the proposed sigmoid-variant model fits well the empirical relationship between dimension and convergence rounds. Furthermore, the proposed optimization scheme has achieved lower total energy consumption than the baseline method, with dimension selection having the most significant impact among the optimization variables. 


%


\ifCLASSOPTIONcaptionsoff
  \newpage
\fi



%

%

\bibliographystyle{IEEEtran}
\bibliography{ref}
\end{document}